\documentclass[a4paper,UKenglish,cleveref, autoref, thm-restate]{lipics-v2021}

\usepackage[linesnumbered,titlenotnumbered,ruled]{algorithm2e}
\usepackage{tikz,pgffor}
\usetikzlibrary{arrows,backgrounds}
\usepackage{tkz-euclide}
\usetikzlibrary{arrows}
\usetikzlibrary{shapes}
\usetikzlibrary{calc}
\usetikzlibrary{automata}
\usetikzlibrary{positioning}
\usepackage{listings}
\usepackage{transparent}
\usepackage{todonotes}
\usepackage{enumitem}

\nolinenumbers

\usepackage{calrsfs}
\DeclareMathAlphabet{\pazocal}{OMS}{zplm}{m}{n}

\newtheorem*{lemma*}{Lemma}
\newtheorem*{disclaimer*}{Disclaimer}

\newcommand{\N}{\mathbb{N}}
\newcommand{\R}{\mathbb{R}}
\newcommand{\Z}{\mathbb{Z}}
\newcommand{\A}{\pazocal{A}}
\newcommand{\C}{\pazocal{C}}

\newcommand{\B}{\mathbb{B}}
\newcommand{\F}{\pazocal{F}}
\newcommand{\T}{\pazocal{T}}

\newcommand{\V}{\pazocal{V}}

\newcommand{\calC}{\pazocal{C}}
\newcommand{\calT}{\pazocal{T}}
\newcommand{\calL}{\pazocal{L}}

\newcommand{\calO}{\pazocal{O}}

\newcommand{\mecs}{\mathit{mecs}}
\newcommand{\MD}{\mathrm{MD}}

\newcommand{\ce}[1]{\left(#1 \right)}

\newcommand{\size}[1]{|\!|#1|\!|}

\renewcommand{\next}{\mathit{next}}
\newcommand{\from}{\mathit{from}}
\newcommand{\effect}{\mathit{effect}}
\newcommand{\RankEff}{\mathit{RankEff}}
\newcommand{\length}{\mathit{len}}

\newcommand{\Decompose}{Decomp}

\newcommand{\prob}{\ensuremath{\mathbb{P}}}

\newcommand{\update}{\ensuremath{\mathbb{U}}}

\newcommand{\conSysA}{\ensuremath{A}}
\newcommand{\conSysB}{\ensuremath{B}}
\newcommand{\conSysC}{\ensuremath{C}}
\newcommand{\conSysD}{\ensuremath{D}}

\newcommand{\vass}{\ensuremath{\mathit{B}}}
\newcommand{\counters}{\ensuremath{x}}
\newcommand{\flowMatrix}{\ensuremath{F}}
\newcommand{\rankCoeff}{\ensuremath{y}}

\newcommand{\vars}{\ensuremath{d}}
\newcommand{\offsets}{\ensuremath{z}}
\newcommand{\states}{\ensuremath{Q}}
\newcommand{\updates}{\ensuremath{U}}

\newcommand{\var}{\ensuremath{c}}

\newcommand{\character}{\ensuremath{\mathtt{char}}}
\newcommand{\identity}{\ensuremath{Id}}

\newcommand{\oneVec}{\ensuremath{\mathbf{1}}}

\newcommand{\transition}{\ensuremath{a}}

\newcommand{\bx}{\mathbf{x}}
\newcommand{\by}{\mathbf{y}}
\newcommand{\bz}{\mathbf{z}}
\newcommand{\bu}{\mathbf{u}}
\newcommand{\bv}{\mathbf{v}}
\newcommand{\bw}{\mathbf{w}}
\newcommand{\bn}{\mathbf{n}}

\newcommand{\bj}{\mathbf{j}}
\newcommand{\bi}{\mathbf{i}}

\newcommand{\Term}{\mathit{Term}}

\newcommand{\Q}{\ensuremath{\mathbb{Q}}}
\newcommand{\Nset}{\ensuremath{\mathbb{N}}}

\newcommand{\PTIME}{\mathbf{P}}

\newcommand{\NP}{\mathbf{NP}}
\newcommand{\coNP}{\mathbf{coNP}}

\newcommand{\Exp}{\ensuremath{\mathbb{E}}}
\newcommand{\conf}{\mathit{C}}
\newcommand{\eexp}{\mathrm{exp}}
\newcommand{\Reach}{\mathrm{Reach}}
\newcommand{\weigth}{\mathit{weight}}
\newcommand{\tin}{\mathit{In}}
\newcommand{\tout}{\mathit{Out}}

\newcommand{\M}{\ensuremath{\mathcal{M}}}

\tikzstyle{loop above}=[tran, to path={.. controls +(60:.5) 
	and +(120:.5) .. (\tikztotarget) \tikztonodes}]
\tikzstyle{loop below}=[tran, to path={.. controls +(240:.5) 
	and +(300:.5) .. (\tikztotarget) \tikztonodes}]
\tikzstyle{loop left}=[tran,  to path={.. controls +(150:.5) 
	and +(210:.5) .. (\tikztotarget) \tikztonodes}]
\tikzstyle{loop right}=[tran,  to path={.. controls +(330:.5) 
	and +(30:.5) .. (\tikztotarget) \tikztonodes}]
\tikzstyle{bigstoch}=[circle,draw,minimum size=8ex,inner sep=0pt,font=\Large, very thick,text centered, fill=blue!20, fill opacity=0.2, draw=black!80, text opacity=1]
\tikzstyle{stoch}=[circle,draw,minimum size=4ex,inner sep=0pt,font=\large, very thick,text centered, fill=blue!20, fill opacity=0.2, draw=black!80, text opacity=1]
\tikzstyle{bigstoch2}=[circle,draw,minimum size=8ex,inner sep=0pt,font=\Large, very thick,text centered, fill=blue!20, fill opacity=0.2, draw=black!80, text opacity=1, double]
\tikzstyle{bigmin}=[diamond,draw,minimum size=8ex,inner sep=0pt,font=\Large, very thick,text centered, fill=blue!20, fill opacity=0.2, draw=black!80, text opacity=1]
\tikzstyle{max}=[rectangle,draw,minimum size=4ex,inner sep=0pt,font=\large, very thick,text centered, fill=blue!20, fill opacity=0.2, draw=black!80, text opacity=1]
\tikzstyle{bigtran}=[very thick,draw,-angle 60,font=\scriptsize, inner sep = 6pt]
\tikzstyle{tran}=[thick,draw,font=\scriptsize, inner sep = 6pt,-stealth]

\title{Asymptotic Complexity Estimates for Probabilistic Programs and their VASS Abstractions} 

\titlerunning{Asymptotic Estimates for VASS MDPs} 

\author{Michal Ajdar{\'{o}}w}{Masaryk University, Czechia \and \url{https://www.muni.cz/lide/422654-michal-ajdarow}}{xajdarow@fi.muni.cz}{https://orcid.org/0000-0003-0694-0944}{}

\author{Anton\'{\i}n Ku\v{c}era}{Masaryk University, Czechia \and \url{https://www.fi.muni.cz/usr/kucera/}}{tony@fi.muni.cz}{https://orcid.org/0000-0002-6602-8028}{}

\authorrunning{M.~Ajdar{\'{o}}w and A.~Ku\v{c}era} 

\Copyright{Michal Ajdar{\'{o}}w and Anton\'{\i}n Ku\v{c}era} 



\category{} 

\funding{The work is supported by the Czech Science Foundation, Grant No.~21-24711S.}

\hideLIPIcs  

\EventNoEds{2}
\EventLongTitle{34th International Conference on Concurrency Theory (CONCUR 2023)}
\EventShortTitle{CONCUR 2023}
\EventAcronym{CONCUR}
\EventYear{2023}
\EventDate{September 18--23, 2023}
\EventLocation{Antwerp, Belgium}
\EventLogo{}

\begin{document}

\maketitle

\begin{abstract}
The standard approach to analyzing the asymptotic complexity of probabilistic programs is based on studying the asymptotic growth of certain expected values (such as the expected termination time) for increasing input size. We argue that this approach is not sufficiently robust, especially in situations when the expectations are infinite. We propose new estimates for the asymptotic analysis of probabilistic programs with non-deterministic choice that overcome this deficiency. Furthermore, we show how to efficiently compute/analyze these estimates for selected classes of programs represented as Markov decision processes over vector addition systems with states.
\end{abstract}

\section{Introduction}
\label{sec-intro}

				
				 
	

Vector Addition Systems with States (VASS) \cite{HP:VASS-reachability-TCS} are a model for discrete systems with multiple unbounded resources expressively equivalent to Petri nets \cite{Petri:first-paper}. Intuitively, a VASS with $d \geq 1$ counters is a finite directed graph where the transitions are labeled by $d$-dimensional vectors of integers representing \emph{counter updates}. A computation starts in some state for some initial vector of non-negative counter values and proceeds by selecting transitions non-deterministically and performing the associated counter updates. Since the counters cannot assume negative values, transitions that would decrease some counter below zero are disabled.

In program analysis, VASS are used as abstractions for programs operating over unbounded integer variables. Input parameters are represented by initial counter values, and more complicated arithmetical functions, such as multiplication, are modeled by VASS gadgets computing these functions in a weak sense (see, e.g., \cite{LS:Petri-computer}). Branching constructs, such as \textbf{if-then-else}, are usually replaced with non-deterministic choice. VASS are particularly useful for evaluating the \emph{asymptotic complexity} of infinite-state programs, i.e., the dependency of the running time (and other complexity measures)  on the size of the program input \cite{SZV:amortized,SZV:difference-constraints}. Traditional VASS decision problems such as reachability, liveness, or boundedness are computationally hard \cite{CLLLM:VASS-reach-nonelem,Lipton:PN-Reachability,MM:containment-Petri}, and other verification problems such as equivalence-checking \cite{Jancar:PN-bisimilarity-TCS} or model-checking \cite{Esparza:ModelChecking-AI} are even undecidable. In contrast to this, decision problems related to the asymptotic growth of VASS complexity measures are solvable with low complexity and sometimes even in \emph{polynomial time} \cite{BCKNVZ:VASS-linear-termination,Zuleger:VASS-polynomial,KLV:VASS-Grzegorczyk,Leroux:Polynomial-termination-VASS,AK:VASS-polynomial-termination}; see \cite{Kucera:Asymptotic-VASS-Analysis-SIGLOG} for a recent overview.

The existing results about VASS asymptotic analysis are applicable to programs with non-determinism (in \emph{demonic} or \emph{angelic} form, see \cite{BW:nondet-languages}), but cannot be used to analyze the complexity of \emph{probabilistic programs}. This motivates the study of Markov decision process over VASS (VASS MDPs) with both non-deterministic and probabilistic states, where transitions in probabilistic states are selected according to fixed probability distributions. Here, the problems of asymptotic complexity analysis become even more challenging because VASS MDPs subsume infinite-state stochastic models that are notoriously hard to analyze. So far, the only existing result about asymptotic VASS MDP analysis is \cite{BCKNV:probVASS-linear-termination} where the linearity of expected termination time is shown decidable in polynomial time for VASS MDPs with DAG-like MEC decomposition.

\textbf{Our Contribution:} We study the problems of asymptotic complexity analysis for probabilistic programs and their VASS abstractions. 

For non-deterministic programs, termination complexity is a function $\calL_{\max}$ assigning to every $n \in \Nset$ the length of the longest computation initiated in a configuration with each counter set to~$n$. A natural way of generalizing this concept to probabilistic programs is to define a function $\calL_{\eexp}$ such that $\calL_{\eexp}(n)$ is the maximal \emph{expected length} of a computation initiated in a configuration of size~$n$, where the maximum is taken over all strategies resolving non-determinism. The same approach is applicable to other complexity measures. We show that this natural idea is generally \emph{inappropriate}, especially in situations when $\calL_{\eexp}(n)$ is \emph{infinite} for a sufficiently large~$n$. By ``inappropriate'' we mean that this form of asymptotic analysis can be misleading. For example, if $\calL_{\eexp}(n) = \infty$ for all $n \geq 1$, one may conclude that the computation takes a very long time independently of~$n$. However, this is not necessarily the case, as demonstrated in a simple example of Fig.~\ref{fig-prob-prg} (we refer to Section~\ref{sec-asymptotic-measures} for a detailed discussion). Therefore, we propose new notions of \emph{lower/upper/tight complexity estimates} and demonstrate their advantages over the expected values. These notions can be adapted to other models of probabilistic programs, and constitute the main conceptual contribution of our work.

Then, we concentrate on algorithmic properties of the complexity estimates in the setting of VASS MDPs. Our first result concerns \emph{counter complexity}. We show that for every VASS MDP with DAG-like MEC decomposition and every counter~$c$, there are only two possibilities:
\begin{itemize}
	\item The function $n$ is a \emph{tight estimate} of the asymptotic growth of the maximal $c$-counter value assumed along a computation initiated in a configuration of size~$n$.
	\item The function $n^2$ is a \emph{lower estimate} of the asymptotic growth of the maximal $c$-counter value assumed along a computation initiated in a configuration of size~$n$.
\end{itemize}
Furthermore, it is decidable in \emph{polynomial time} which of these alternatives holds. 

Since the termination and transition complexities can be easily encoded as the counter complexity for a fresh ``step counter'', the above result immediately extends also to these complexities. To some extent, this result can be seen as a generalization of the result about termination complexity presented in \cite{BCKNV:probVASS-linear-termination}. See Section~\ref{section-linquad} for more details.

Our next result is a full classification of asymptotic complexity for one-dimensional VASS MDPs. We show that for every one-dimensional VASS MDP 
\begin{itemize}
	\item the counter complexity is either unbounded or $n$ is a tight estimate;
	\item termination complexity is either unbounded or one of the functions $n$, $n^2$ is a tight estimate.
	\item transition complexity is either unbounded, or bounded by a constant, or one of the functions $n$, $n^2$ is a tight estimate.
\end{itemize}
Furthermore, it is decidable in \emph{polynomial time} which of the above cases hold.

Since the complexity of the considered problems remains low, the results are encouraging. On the other hand, they require non-trivial insights, indicating that establishing a full and effective classification of the asymptotic complexity of multi-dimensional VASS MDPs is a challenging problem.

\section{Preliminaries}
\label{sec-prelim}


We use $\N$, $\Z$, $\Q$, and $\R$ to denote the sets of non-negative integers,
integers, rational numbers, and real numbers. 
%
Given a function $f \colon \N \rightarrow \N$, we use  $O(f)$ and $\Omega(f)$ to denote the sets of all $g \colon \N \rightarrow \N$ such that $g(n) \leq a \cdot f(n)$ and $g(n) \geq b \cdot f(n)$ for all sufficiently large $n \in \N$, where $a,b$ are some positive constants. If $h \in O(f)$ and $h \in \Omega(f)$, we write $h \in \Theta(f)$.

Let $A$ be a finite index set. The vectors of $\R^A$ are denoted by bold letters such as $\bu,\bv,\bz,\ldots$. The component of $\bv$ of index $i\in A$ is denoted by $\bv(i)$. 
If the index set is of the form $A=\{1,2,\dots,d\}$ for some positive integer $d$, we write $\R^d$ instead of $\R^A$. For every $n \in \N$, we use $\bn$ to denote the constant vector where all
components are equal to~$n$.
 The other
standard operations and relations on $\R$ such as
$+$, $\leq$, or $<$ are extended to $\R^d$ in the component-wise way. In particular,  $\bv < \bu$ if $\bv(i) < \bu(i)$ for every index $i$.

A \emph{probability distribution} over a finite set $A$ is a vector $\nu \in [0,1]^A$ such that $\sum_{a \in A}\nu(a) = 1$. We say that $\nu$ is \emph{rational} if every $\nu(a)$ is rational, and $\emph{Dirac}$ if $\nu(a) =1$ for some $a \in A$.
\vspace{-0.2cm}

\subsection{VASS Markov Decision Processes}

\begin{definition}
	\label{def-VASS} 
	
	Let $d \geq 1$. A \emph{$d$-dimensional VASS MDP} is a tuple $\A = \ce{Q, (Q_n,Q_p),T,P}$, where 
	\begin{itemize}
		\item $Q \neq \emptyset$ is a finite set of \emph{states} split into two disjoint subsets $Q_n$ and $Q_p$ of \emph{nondeterministic} and \emph{probabilistic} states,
		\item $T \subseteq Q \times \Z^d\times Q$ is a finite set of \emph{transitions} such that, for every $p \in Q$, the set $\tout(p) \subseteq T$ of all transitions of the form $(p,\bu,q)$ is non-empty.
		\item $P$ is a function assigning to each $t \in \tout(p)$ where $p \in Q_p$ a positive rational probability so that  $\sum_{t \in T(p)} P(t) =1$. 
	\end{itemize}
\end{definition}
The encoding size of $\A$ is denoted by $\size{\A}$, where the integers representing counter updates are written in binary and probability values are written as fractions of binary numbers. For every $p \in Q$, we use $\tin(p) \subseteq T$ to denote the set of all transitions of the form $(q,\bu,p)$. The update vector $\bu$ of a transition $t = (p,\bu,q)$ is also denoted by $\bu_t$.

A \emph{finite path} in $\A$ of length~$n \geq 0$ is a finite sequence of the form
$p_0,\bu_1,p_1,\bu_2,\ldots,\bu_{n},p_n$ where $(p_i,\bu_{i+1},p_{i+1}) \in T$  for all $i<n$. We use $\length(\alpha)$ to denote the length of $\alpha$. If there is a finite path from $p$ to $q$, we say that $q$ is \emph{reachable} from~$p$.
An \emph{infinite path} in $\A$ is an infinite sequence $\pi = p_0,\bu_1,p_1,\bu_2,\ldots$ such that every finite prefix of $\pi$ ending in a state is a finite path in~$\A$. 


A \emph{strategy} is a function $\sigma$ assigning to every finite path $p_0,\bu_1,\ldots,p_n$ such that $p_n \in Q_n$ a probability distribution over~$\tout(p_n)$. A strategy is \emph{Markovian (M)} if it depends only on the last state $p_n$, and \emph{deterministic (D)} if it always returns a Dirac distribution. The set of all strategies is denoted by $\Sigma_{\A}$, or just $\Sigma$ when $\A$ is understood. Every initial state $p \in Q$ and every strategy $\sigma$ determine the probability space over infinite paths initiated in $p$ in the standard way. We use $\prob^\sigma_{p}$ to denote the associated probability measure.

A \emph{configuration} of $\A$ is a pair $p\bv$, where $p \in Q$ and $\bv \in \Z^d$. If some component of $\bv$ is negative, then $p\bv$ is \emph{terminal}. The set of all configurations of $\A$ is denoted by $\conf(\A)$. 

Every infinite path $p_0,\bu_1,p_1,\bu_2,\ldots$ and every initial vector $\bv \in \Z^d$ determine the corresponding \emph{computation} of $\V$, i.e., the sequence of configurations $p_0\bv_0, p_1 \bv_1, p_2 \bv_2,\ldots$ such that $\bv_0 =\bv$ and $\bv_{i+1} = \bv_i + \bu_{i+1}$. Let $\Term(\pi)$ be the least $j$ such that $p_j\bv_j$ is terminal. If there is no such $j$, we put  $\Term(\pi) = \infty$ . 

Note that every computation uniquely determines its underlying infinite path. We define the probability space over all computations initiated in a given $p \bv$, where the underlying probability measure  $\prob^\sigma_{p \bv}$ is obtained from $\prob_p^\sigma$ in an obvious way. For a measurable function $X$ over computations, we use $\Exp^\sigma_{p \bv}[X]$ to denote the expected value of~$X$.

\section{Asymptotic Complexity Measures for VASS MDPs}
\label{sec-asymptotic-measures}

In this section, we introduce asymptotic complexity estimates applicable to probabilistic programs with non-determinism and their abstract models (such as VASS MDPs). We also explain their relationship to the standard measures based on the expected values of relevant random variables.

\begin{figure}
\parbox[c]{.5\textwidth}{
\begin{tabbing}
	\hspace*{1em} \= \hspace*{2em} \= \hspace*{4em} \= \hspace*{4em} \= \kill
	\> \textbf{input} $N$\\
	\> \textbf{repeat}\\
	\>\> \textbf{random choice:}\\
	\>\>\> $0.5:\ \ N:=N+1;$\\
	\>\>\> $0.5:\ \ N:=N-1;$\\
	\> \textbf{until} $N=0$
\end{tabbing}}\hspace*{4em}
\parbox[c]{.5\textwidth}{
	\begin{tikzpicture}[x=2cm, y=2cm]
		\node [stoch,label={[shift={(0,-1)}]$\A$}] (r) at (3,0)  {$p$};
		\draw [tran,loop right] (r) to node[right] {$0.5, +1$} (r); 
		\draw [tran,loop left]  (r) to node[left] {$0.5, -1$} (r); 
	\end{tikzpicture}
}
\caption{A probabilistic program with infinite expected running time for every $N\geq 1$, and its $1$-dimensional VASS MDP model~$\A$.}
\label{fig-prob-prg}
\end{figure}

Let us start with a simple motivating example. Consider the simple probabilistic program of Fig.~\ref{fig-prob-prg}. The program inputs a positive integer $N$ and then repeatedly increments/decrements~$N$ with probability $0.5$ until $N=0$. One can easily show that for every $N \geq 1$, the program terminates with probability one, and the expected termination time is \emph{infinite}. Based on this, one may conclude that the execution takes a very long time, independently of the initial value of~$N$. However, this conclusion is \emph{not} consistent with practical experience gained from trial runs\footnote{For $N=1$, about $95\%$ of trial runs terminate after at most $1000$ iterations of the \textbf{repeat-until} loop.  For $N=10$, only about $75\%$ of all runs terminate after at most $1000$ iterations, but about 
$90\%$ of them terminate after at most $10000$ iterations.}. The program tends to terminate ``relatively quickly'' for small~$N$, and the termination time \emph{does} depend on~$N$. 
Hence, the function assigning $\infty$ to every $N \geq 1$ is \emph{not} a faithful characterization of the asymptotic growth of termination time.  We propose an alternative characterization  based on the following observations\footnote{Formal proofs of these observations are simple; in Section~\ref{sec-one-dim}, we give a full classification of the asymptotic behaviour of one-dimensional VASS MDPs subsuming the trivial example of Fig.~\ref{fig-prob-prg}.}:
\begin{itemize}
		\item For every $\varepsilon >0$, the probability of all runs terminating after more than $n^{2+\varepsilon}$ steps (where $n$ is the initial value of $N$) approaches \emph{zero} as $n \to \infty$.
		\item For every $\varepsilon >0$, the probability of all runs terminating after more than $n^{2-\varepsilon}$ steps (where $n$ is the initial value of $N$) approaches \emph{one} as $n \to \infty$.
\end{itemize}
Since the execution time is ``squeezed'' between $n^{2-\varepsilon}$ and $n^{2+\varepsilon}$ for an arbitrarily small $\varepsilon > 0$ as $n \to \infty$, it can be characterized as ``asymptotically quadratic''. This analysis is in accordance with experimental outcomes.

\subsection{Complexity of VASS Runs}
\label{sec-comp-VASS-runs}
We recall the complexity measures for VASS runs used in previous works \cite{BCKNVZ:VASS-linear-termination,Zuleger:VASS-polynomial,KLV:VASS-Grzegorczyk,Leroux:Polynomial-termination-VASS,AK:VASS-polynomial-termination}. These functions can be seen as variants of the standard time/space complexities for Turing machines.  

Let $\A = \ce{Q, (Q_n,Q_p),T,P}$ be a $d$-dimensional VASS MDP, $c \in \{1,\ldots,d\}$, and $t \in T$. For every computation  $\pi = p_0 \bv_0,p_1 \bv_1,p_2 \bv_2,\ldots$, we put
\begin{eqnarray*}
    \calL(\pi)      & = & \Term(\pi)\\
	\calC[c](\pi) & = & \sup \{\bv_i(c) \mid 0 \leq i < \Term(\pi)\}\\
    \calT[t](\pi) & = & \mbox{the total number of all $0 \leq i < \Term(\pi)$ such that $(p_i,\bv_{i{+}1}{-}\bv_i,p_{i+1}) = t$}
\end{eqnarray*}
We refer to the functions $\calL$, $\calC[c]$, and $\calT[t]$ as \emph{termination}, \emph{$c$-counter}, and \emph{$t$-transition complexity}, respectively.

Let $\F$ be one of the complexity functions defined above. In VASS abstractions of computer programs, the input is represented by initial counter values, and the input size corresponds to the maximal initial counter value. The existing works on \emph{non-probabilistic} VASS concentrate on analyzing the asymptotic growth of the functions $\F_{\max} : \Nset \to \Nset_\infty$ where
\begin{eqnarray*}
    \F_{\max}(n)     & = & \max\{\F(\pi) \mid \pi \mbox { is a computation initiated in $p\bn$ where $p \in Q$}\}
\end{eqnarray*}
For VASS MDP, we can generalize $\F_{\max}$ into $\F_{\eexp}$ as follows:
\begin{eqnarray*}
    \F_{\eexp}(n)     & = & \max\{\Exp_{p\bn}^{\sigma}[\F] \mid \sigma \in \Sigma_{\A}, p \in Q\}
\end{eqnarray*}
Note that for non-probabilistic VASS, the values of $\F_{\max}(n)$ and $\F_{\eexp}(n)$ are the same. However, the function $\F_{\eexp}$ suffers from the deficiency illustrated in the motivating example at the beginning of Section~\ref{sec-asymptotic-measures}. To see this, consider the one-dimensional VASS MDP $\A$ modeling the simple probabilistic program (see Fig.~\ref{fig-prob-prg}). For every $n \geq 1$ and the only (trivial) strategy $\sigma$, we have that $\prob^\sigma_{p \bn}[\Term < \infty] = 1$ and $\calL_{\eexp}(n) = \infty$. However, the practical experience with trial runs of $\A$ is the same as with the original probabilistic program (see above).

\subsection{Asymptotic Complexity Estimates}
\label{sec-estimates}

In this section, we introduce asymptotic complexity estimates allowing for a precise analysis of the asymptotic growth of the termination, $c$-counter, and $t$-transition complexity, especially when their expected values are infinite for a sufficiently large input. For the sake of readability, we first present a simplified variant applicable to \emph{strongly connected} VASS MDPs.

Let $\F$ be one of the complexity functions for VASS computations defined in Section~\ref{sec-comp-VASS-runs}, and let $f : \Nset \to \Nset$. We say that $f$ is a \emph{tight estimate of $\F$} if, for arbitrarily small $\varepsilon > 0$, the value of $\F(n)$ is ``squeezed'' between $f^{1-\varepsilon}(n)$ and $f^{1+\varepsilon}(n)$ as $n \to \infty$. More precisely, for every $\varepsilon > 0$,

\begin{itemize}
	\item there exist $p \in Q$ and strategies $\sigma_1,\sigma_2,\ldots$ such that $\liminf_{n \to \infty} 	\ \prob^{\sigma_n}_{p \bn}[\F \geq (f(n))^{1-\varepsilon}] \ = \ 1$;
	\item for all $p \in Q$ and strategies $\sigma_1,\sigma_2,\ldots$ we have that $\limsup_{n \to \infty} 	\ \prob^{\sigma_n}_{p \bn}[\F \geq (f(n))^{1+\varepsilon}] \ = \ 0$.
\end{itemize}

The above definition is adequate for strongly connected VASS MDPs because tight estimates tend to exist in this subclass. Despite some effort, we have not managed to construct an example of a strongly connected VASS MDP where an $\F$ with some upper polynomial estimate does \emph{not} have a tight estimate (see Conjecture~\ref{con-estim}).  However, if the underlying graph of $\A$ is \emph{not} strongly connected, then the asymptotic growth of $\F$ can differ for computations visiting a different sequence of maximal end components (MECs) of $\A$, and the asymptotic growth of $\F$ can be ``squeezed'' between $f^{1-\varepsilon}(n)$ and $f^{1+\varepsilon}(n)$ only for the subset of computations visiting the same sequence of MECs. This explains why we need a more general definition of complexity estimates presented below.

An \emph{end component (EC)} of $\A$ is a pair $(C,L)$ where $C \subseteq Q$ and  $L \subseteq T$ such that the following conditions are satisfied:
\begin{itemize}
	\item $C \neq \emptyset$;
	\item if $p \in C \cap Q_n$, then at least one outgoing transition of $p$ belongs to $L$;
	\item if $p \in C \cap Q_p$, then all outgoing transitions of $p$ belong to $L$;
	\item if $(p,\bu,q) \in L$, then $p,q \in C$;
	\item for all $p,q \in C$ we have that $q$ is reachable from $p$ and vice versa.
\end{itemize} 
Note that if $(C,L)$ and $(C',L')$ are ECs such that $C \cap C' \neq \emptyset$, then $(C\cup C', L \cup L')$ is also an EC. Hence, every $p\in Q$ either belongs to a unique \emph{maximal end component} (MEC), or does not belong to any EC. Also observe that each MEC can be seen as a strongly connected VASS MDP. We say that $\A$ has \emph{DAG-like MEC decomposition} if for every pair $M,M'$ of different MECs such that the states of $M'$ are reachable from the states of $M$ we have that the states of $M$ are not reachable from the states of~$M'$.  

For every infinite path $\pi$ of $\A$, let $\mecs(\pi)$ be the unique sequence of MECs visited by~$\pi$. Observe that $\mecs(\pi)$ disregards the states that do not belong to any EC; intuitively, this is because the transitions executed in such states do not influence the asymptotic growth of~$\F$. Observe that the length of $\mecs(\pi)$, denoted by $\mathit{len}(\mecs(\pi))$, can be finite or infinite. The first possibility corresponds to the situation when an infinite suffix of $\pi$ stays within the same MEC. Furthermore, for all $\sigma \in \Sigma$ and $p \in Q$, we have that $\prob^\sigma_{p}[\mathit{len}(\mecs) = \infty] =0$, and the probability $\prob^\sigma_{p}[\mathit{len}(\mecs) \geq k]$ decays exponentially in~$k$ (these folklore results are easy to prove). All of these notions are lifted to computations in an obvious way.

Observe that if a strategy $\sigma$ aims at maximizing the growth of $\F$, we can safely assume that $\sigma$ eventually stays in a \emph{bottom} MEC that cannot be exited (intuitively, $\sigma$ can always move from a non-bottom MEC to a bottom MEC by executing a few extra transitions that do not influence the asymptotic growth of~$\F$, and the bottom MEC may allow increasing $\F$ even further). On the other hand, the maximal asymptotic growth of $\F$ may be achievable along some ``minimal'' sequence of MECs, and this information is certainly relevant for understanding the behaviour of a given probabilistic program. This leads to the following definition: 


\begin{definition}
\label{def-estimates}
	A \emph{type} is a finite sequence $\beta$ of MECs such that $\mecs(\pi) = \beta$ for some infinite path~$\pi$.

	We say that $f$ is a \emph{lower estimate of $\F$ for a type $\beta$} if for every $\varepsilon > 0$ there exist $p \in Q$ and a sequence of strategies $\sigma_1,\sigma_2,\ldots$ such that $\prob^{\sigma_n}_{p\bn}[\mecs = \beta] > 0$ for all $n \geq 1$ and
	\[
   		\liminf_{n \to \infty} \ \prob^{\sigma_n}_{p \bn}[\F \geq (f(n))^{1-\varepsilon} \mid \mecs{=}\beta] \ = \ 1\,.
	\]
	Similarly, we say that $f$ is an \emph{upper estimate of $\F$ for a type $\beta$} if for every $\varepsilon>0$, every $p \in Q$, and every sequence of strategies $\sigma_1,\sigma_2,\ldots$ such that $\prob^{\sigma_n}_{p\bn}[\mecs = \beta] > 0$ for all $n \geq 1$ we have that
	\[
	\limsup_{n \to \infty} 	\ \prob^{\sigma_n}_{p \bn}[\F \geq (f(n))^{1+\varepsilon} \mid \mecs{=}\beta] \ = \ 0
	\]	
	If there is no upper estimate of $\F$ for a type $\beta$, we say that $\F$ is \emph{unbounded for $\beta$}. Finally, we say that $f$ is a \emph{tight estimate of $\F$ for $\beta$} if it is both a lower estimate and an upper estimate of $\F$ for $\beta$.
\end{definition}


Let us note that in the subclass of \emph{non-probabilistic} VASS, MECs become strongly connected components (SCCs), and types correspond to paths in the directed acyclic graph of SCCs. Each such path determines the corresponding asymptotic increase of $\F$, as demonstrated in \cite{AK:VASS-polynomial-termination}. We conjecture that types play a similar role for VASS MDPs. More precisely, we conjecture the following:

\begin{conjecture}
\label{con-estim}
   If some polynomial is an upper estimate of $\F$ for $\beta$, then there exists a tight estimate $f$ of $\F$ for~$\beta$. 
\end{conjecture}
Even if Conjecture~\ref{con-estim} is proven wrong, there are interesting subclasses of VASS MDPs where it holds, as demonstrated in subsequent sections.

\begin{figure}\centering
\begin{tikzpicture}[x=3cm, y=3cm]
	\node [max]   (a) at (0,0)   {};
	\node [stoch] (b) at (1,0)   {};
	\node [max]   (c) at (-2,-1) {};
	\node [stoch] (d) at (0,-1)  {$q$};
	\node [max]   (e) at (-1,-1)  {};
	\node [stoch] (f) at (1,-1)  {};
	\draw[tran] (b) to [bend right=50] node[above=-1ex] {$\frac{1}{2}, (-1,+1,0,)$} (a);
	\draw[tran] (b) to [bend left=50]  node[below=-1ex] {$\frac{1}{2}, (+1,+1,0,)$} (a);
	\draw[tran] (a) to  node[above=-1ex]  {$\vec{0}$} (b);
	\draw[tran,rounded corners] (a) -- node[above=-1ex] {$\vec{0}$} +(-2,0)  -- (c);
	\draw[tran] (a) to  node[left=-1ex, near end]  {$\vec{0}$} (d);
	\draw[tran] (d) to  node[above=-1ex]  {$\frac{1}{2}, \vec{0}$} (e);
	\draw[tran] (d) to  node[above=-1ex]  {$\frac{1}{2}, \vec{0}$} (f);
	\draw[tran,loop below]  (e) to  node[below=-1ex]  {$(0,-1,0)$} (e);
	\draw[tran,loop below]  (c) to  node[below=-1ex]  {$(0,-1,+1)$} (c);
	\draw[tran,loop below]  (f) to  node[below=-1ex]  {$\frac{1}{2},(0,-1,+1)$} (f);
	\draw[tran,loop right]  (f) to  node[right=-1ex]  {$\frac{1}{2},(0,+1,+1)$} (f);
	\draw[dotted,thick, rounded corners] (-.2,-.55) rectangle (1.4,0.55);
	\draw[dotted,thick, rounded corners] (-1.3,-1.55) rectangle (-0.7,-.55);
	\draw[dotted,thick, rounded corners] (-2.3,-1.55) rectangle (-1.7,-.55);
	\draw[dotted,thick, rounded corners] (.65,-1.55) rectangle (2,-.65);
    \node at (1.3,0.45) {$M_1$};
    \node at (-0.8,-0.65) {$M_3$};
    \node at (-1.8,-0.65) {$M_2$};
    \node at (1.9,-0.75) {$M_4$};
\end{tikzpicture}
\caption{A VASS MDP $\A$ with four MECs and seven types.}
\label{fig-VASS-types}
\end{figure}

For every pair of MECs $M,M'$, let $P(M,M')$ be the maximal probability (achievable by some strategy) of reaching a state of $M'$ from a state of $M$ in $\A$ without passing through a state of some other MEC $M''$. Note that $P(M,M')$ is efficiently computable by standard methods for finite-state MDPs. The \emph{weight} of a given type $\beta = M_1,\ldots,M_k$ is defined as $\weigth(\beta) = \prod_{i=1}^{k-1} P(M_i,M_{i+1})$. Intuitively, $\weigth(\beta)$ corresponds to the maximal probability of ``enforcing'' the asymptotic growth of $\F$ according to the tight estimate $f$ of $\F$ for $\beta$ achievable by some strategy.   

Generally, higher asymptotic growth of $\F$ may be achievable for types with smaller weights. Consider the following example to understand better the types, their weights, and the associated tight estimates. 

\begin{example}
Let $\A$ be the VASS~MDP of Fig.~\ref{fig-VASS-types}. There are four MECs $M_1,M_2,M_3,M_4$ where $M_2,M_3,M_4$ are bottom MECs. Hence, there are four types of length one and three types of length two. Let us examine the types of length two initiated in $M_1$
for $\F \equiv \calC[c]$ where $c$ is the third counter.    

Note that in $M_1$, the first counter is repeatedly incremented/decremented with the same probability~$\frac{1}{2}$. The second counter ``counts'' these transitions and thus it is ``pumped'' to a \emph{quadratic} value (cf.{} the VASS MDP of Fig.~\ref{fig-prob-prg}). Then, a strategy may decide to move to $M_2$, where the value of the second counter is transferred to the third counter. Hence, $n^2$ is the tight estimate of $\calC[c]$ for the type $M_1,M_2$, and $\weigth(M_1,M_2) = 1$. Alternatively, a strategy may decide to move to the probabilistic state~$q$. Then, either $M_3$ or $M_4$ is entered with the same probability~$\frac{1}{2}$, which implies $\weigth(M_1,M_3) = \weigth(M_1,M_4) = \frac{1}{2}$. In $M_3$, the third counter is unchanged, and hence $n$ is the tight estimate of $\calC[c]$ for the type $M_1,M_3$. However, in $M_4$, the second counter previously pumped to a quadratic value is repeatedly incremented/decremented with the same probability $\frac{1}{2}$, and the third counter ``counts'' these transitions. This means that $n^4$ is a tight estimate of $\calC[c]$ for the type~$M_1,M_4$.

This analysis provides detailed information about the asymptotic growth of $\calC[c]$ in~$\A$. Every type shows ``how'' the growth specified by the corresponding tight estimate is achievable, and its weight corresponds to the ``maximal achievable probability of this growth''. This information is completely lost when analyzing the maximal expected value of $\calC[c]$ for computations initiated in configurations $p\bn$ where $p$ is a state of $M_1$, because these expectations are \emph{infinite} for all $n \geq 1$.
\end{example}

Finally, let us clarify the relationship between the lower/upper estimates of $\F$ and the asymptotic growth of $\F_{\eexp}$. The following observation is easy to prove.

\begin{observation}\label{observation-upper-expected-implies-upper-estimate}
	If $\F_{\eexp} \in O(f)$ where $f : \Nset \to \Nset$ is an unbounded function, then $f$ is an upper estimate of $\F$ for every type. Furthermore, if $f : \Nset \to \Nset$ is a lower estimate of $\F$ for some type, then  $\F_{\eexp} \in \Omega(f^{1-\epsilon})$ for each \(\epsilon>0 \). However, if $\F_{\eexp} \in \Omega(f)$ where $f : \Nset \to \Nset$, then $f$ is \emph{not} necessarily a lower estimate of $\F$ for some type.
\end{observation}

Observation~\ref{observation-upper-expected-implies-upper-estimate} shows that complexity estimates are generally more informative than the asymptotics of $\F_{\eexp}$ even if $\F_{\eexp} \in \Theta(f)$ for some ``reasonable'' function~$f$. For example, it may happen that there are only two types $\beta_1$ and $\beta_2$ where $n$ and $n^3$ are tight estimates of $\calL$ for $\beta_1$ and $\beta_2$ with weights $0.99$ and $0.01$, respectively. In this case, $\calL_{\eexp} \in \Theta(n^3)$, although the termination time is linear for $99\%$ of computations.

%
%
%
%
%
%
%

\section{A Dichotomy between Linear and Quadratic Estimates}
\label{section-linquad}

In this section, we prove the following result:

\begin{theorem}
\label{thm-dichotomy}
	Let $\A$ be a VASS MDP with DAG-like MEC decomposition and $\F$ one of the complexity functions $\calL$, $\calC[c]$, or $\calT[t]$. For every type $\beta$, we have that either $n$ is a tight estimate of $\F$ for $\beta$, or $n^2$ is a lower estimate of $\F$ for $\beta$. It is decidable in polynomial time which of the two cases holds.
\end{theorem}

Theorem~\ref{thm-dichotomy} can be seen as a generalization of the linear/quadratic dichotomy results previously achieved for non-deterministic VASS \cite{BCKNVZ:VASS-linear-termination} and for the termination complexity in VASS MDPs \cite{BCKNV:probVASS-linear-termination}. 

It suffices to prove Theorem~\ref{thm-dichotomy} for the \emph{counter complexity}. The corresponding results for the termination and transition complexities then follow as simple consequences. To see this, observe that we can extend a given VASS MDP with a fresh ``step counter'' $sc$ that is incremented by every transition (in the case of $\calL$) or the transition $t$ (in the case of $\T[t]$) and thus ``emulate''  $\calL$ and  $\T[t]$ as $\calC[sc]$.

We first consider the case when $\A$ is strongly connected and then generalize the obtained results to VASS MDPs with DAG-like MEC decomposition. So, let $\A$ be a strongly connected $d$-dimensional VASS MDP and~$c$ a counter of~$\A$. The starting point of our analysis is the dual constraint system designed in \cite{Zuleger:VASS-polynomial} for non-probabilistic strongly connected VASS. We generalize this system to strongly connected VASS MDPs in the way shown in Figure~\ref{fig-systems} (the original system of \cite{Zuleger:VASS-polynomial} can be recovered by disregarding the probabilistic states). 


\begin{figure}[h]
\vspace{0.7cm}
\begin{tabular}{|c|c|}
	\hline
	{\begin{minipage}[c]{0.38\textwidth}\small
			\vspace{0.2cm}
			Constraint system~(I):
			
			\vspace{0.2cm}
			Find $\bx \in \mathbb{Z}^{T}$ such that
			\begin{align}
				\sum_{t \in T} \bx(t) \bu_t  & \ge \vec{0} \nonumber\\
				\bx & \ge \vec{0} \nonumber
			\end{align}
			
		 and for each $p\in Q$
			\begin{align}
				\sum_{t\in \tout(p)}\bx(t)&=\sum_{t \in \tin(p)} \bx(t)\nonumber
		\end{align}
%
%
%
			and for all $p\in Q_p$, $t\in \tout(p)$
			\begin{align}
				\bx(t) & =P(t) \cdot \sum_{t'\in \tout(p)} \bx(t') \nonumber
			\end{align}
			\vspace*{1em}

			\textbf{Objective:} \emph{Maximize}\\
			\begin{itemize}
				\item the number of valid inequalities of the form 
				\[\sum_{t \in T} \bx(t)\bu_t(c) > 0,\]
				\item the number of valid inequalities of the form $\bx(t) > 0$.
			\end{itemize}
			\vspace*{8em}
	\end{minipage}}
	&
	{\begin{minipage}[c]{0.55\textwidth}
			\vspace{0.2cm}
			Constraint system (II):
			
			\vspace{0.2cm}
			Find $\by \in \mathbb{Z}^d,\bz \in \mathbb{Z}^{Q}$ such that
			\begin{align}
				\by & \ge  \vec{0} \nonumber\\
				\bz & \ge  \vec{0} \nonumber
			\end{align}
			and for each $(p,\bu,q)\in T$ where $p\in Q_n$ 			
			   \[\bz(q)-\bz(p)+\sum_{i=1}^{d} \bu(i)\by(i) \leq 0\]
			and for each $p\in Q_p$  
			\[\sum_{t= (p,\bu,q) \in \tout(p)}P(t)\big(\bz(q)-\bz(p)+\sum_{i=1}^{d}  \bu_t(i)\by(i)\big)\leq 0 \]
			\vspace*{1em}
			
			\textbf{Objective:}
			\emph{Maximize}\\
			\begin{itemize}
				\item the number of valid inequalities of the form $\by(c) > 0$,
				\item the number of transitions $t=(p,\bu,q)$ such that $p\in Q_n$ and \[\bz(q)-\bz(p)+\sum_{i=1}^{d} \bu(i)\by(i)<0,\]
				\item the number of states $p\in Q_p$ such that 
				\[\sum_{t = (p,\bu,q) \in \tout(p)}P(t)\big(\bz(q)-\bz(p)+\sum_{i=1}^{d}  \bu(i)\by(i)\big)< 0\,.\]
			\end{itemize}
			\vspace*{1em}
	\end{minipage}}\\
	\hline
\end{tabular}
\caption{Constraint systems for strongly connected VASS MDPs.}
\label{fig-systems}
\end{figure}

%
%
%
%
%
%


Note that solutions of both~(I) and~(II) are closed under addition. Therefore, both~(I) and~(II) have solutions maximizing the specified objectives, computable in polynomial time.
%
%
For clarity, let us first discuss an intuitive interpretation of these solutions, starting with simplified variants obtained for non-probabilistic VASS in \cite{Zuleger:VASS-polynomial}. 

In the non-probabilistic case, a solution of~(I) can be interpreted as a \emph{weighted multicycle}, i.e., as a collection of cycles $M_1,\dots, M_k$ together with weights $a_1,\ldots ,a_k$ such that the total effect of the multicycle, defined by $\sum_{i=1}^k a_i \cdot \mathit{effect}(M_i)$, is non-negative for every counter.  Here, $\mathit{effect}(M_i)$ is the effect of $M_i$ on the counters. The objective of~(I) ensures that the multicycle includes as many transitions as possible, and the total effect of the multicycle is positive on as many counters as possible. For VASS MDPs, the $M_1,\dots, M_k$ should not be interpreted as cycles but as Markovian strategies for some ECs, and $\mathit{effect}(M_i)$ corresponds to the vector of expected counter changes per transition in~$M_i$. The objective of~(I) then maximizes the number of transitions used in the strategies $M_1,\dots, M_k$, and the number of counters where the expected effect of the ``multicycle'' is positive.


A solution of~(II) for non-probabilistic VASS can be interpreted as a ranking function for configurations defined by $\mathit{rank}(p\bv)=\bz(p)+\sum_{i=1}^{d} \by(i)\bv(i)$, such that the value of $\mathit{rank}$ cannot increase when moving from a configuration $p\bv$ to a configuration $q\bu$ using a transition $t=(p,\bu-\bv,q)$. The objective of~(II) ensures that as many transitions as possible decrease the value of $\mathit{rank}$, and $\mathit{rank}$  depends on as many counters as possible. For VASS MDPs, this interpretation changes only for the outgoing transitions $t=(p,\bu,q)$ of probabilistic states. Instead of considering the change of $\mathit{rank}$ caused by such $t$,  we now consider the expected change of $\mathit{rank}$ caused by executing a step from $p$. The objective ensures that $\mathit{rank}$ depends on as many counters as possible, the value of $\mathit{rank}$ is decreased by as many outgoing transitions of non-deterministic states as possible, and the expected change of $\mathit{rank}$ caused by performing an step is negative in as many probabilistic states as possible.

The key tool for our analysis is the following dichotomy (a proof is in Appendix~\ref{app-lem:optimization-duality-proof}).

\begin{lemma}\label{lemma:dichotomy}
	Let $\bx$ be a (maximal) solution to the constraint system~(I) and $\by,\bz$ be a (maximal) solution to the constraint system~(II). Then, for each counter $c$ we have that either $\by(c)>0$ or $\sum_{t \in T} \bx(t)\bu_t(c)>0$, and for each transition $t = (p,\bu,q)\in T$ we have that
	\begin{itemize}
		\item if $p\in Q_n$ then either $\bz(q)-\bz(p)+\sum_{i=1}^{d} \bu(i)\by(i)<0$ or $\bx(t)>0$;
		\item if $p\in Q_p$ then either \[\sum_{t'=(p,\bu',q') \in \tout(p)}P(t')\big(\bz(q')-\bz(p)+\sum_{i=1}^{d}  \bu'(i)\by(i)\big)< 0\] or $\bx(t)>0$.		
	\end{itemize}
	
	
\end{lemma}

For the rest of this section, we fix a maximal solution $\bx$ of~(I) and a maximal solution $\by,\bz$ of~(II), such that the smallest non-zero element of $\by,\bz$ is at least~$1$. We define a ranking function $\mathit{rank}: \conf(\A)\rightarrow \Nset$ as $\mathit{rank}(s\bv)=\bz(s)+\sum_{i=1}^{d} \bv(i)\by(i)$. Now we prove the following theorem:

\begin{theorem}\label{lemma-C-limPlinear-if-no-+}
	For each counter $c$, if $\by(c)>0$ then $n$ is a tight estimate of $\calC[c]$ (for the only type of~$\A$).  Otherwise, i.e., when $\by(c)=0$, the function $n^2$ is a lower estimate of $\calC[c]$.
\end{theorem}
Note that Theorem~\ref{lemma-C-limPlinear-if-no-+} implies Theorem~\ref{thm-dichotomy} for strongly connected VASS MDPs. A proof is obtained by combining the following lemmata.  

\begin{lemma}\label{lemma-linear-upper-bound}
	For every counter $c$ such that $\by(c)>0$, every $\varepsilon >0$, every $p\in Q$, and every $\sigma \in \Sigma$,  there exists $n_0$ such that for all $n\geq n_0$ we have that $\prob^{\sigma}_{p \bn}(\calC[c] \geq n^{1+\varepsilon} )\leq kn^{-\varepsilon}$ where~$k$ is a constant depending only on~$\A$. 
\end{lemma}
A proof is in Appendix~\ref{app-lem:linear-upper-bound-proof}.
 
For $\mathit{Targets} \subseteq \conf(\A)$ and $m \in \Nset$, we use $\Reach^{\leq m}(\mathit{Targets})$ to denote the set of all computations $\pi = p_0\bv_0, p_1\bv_1,\ldots$ such that $p_i\bv_i \in \mathit{Targets}$ for some $i \leq m$.

\begin{lemma}
	\label{lemma-quadratic-lower-bound}
	For each counter $c$ such that $\by(c)=0$ we have that $\calC_{\eexp}[c] \in \Omega(n^2)$ and $n^2$ is a lower estimate of $\calC[c]$. Furthermore, for every $\varepsilon>0$ there exist a sequence of strategies $\sigma_1,\sigma_2,\dots $, a constant $k$, and  $p\in Q$ such that for every $0<\varepsilon'<\varepsilon$, we have that 
	\[
		\lim_{n \to \infty} \prob_{p\bn}^{\sigma_n}(\Reach^{\leq kn^{2-\varepsilon'}}(\mathit{Targets}_n)) = 1
	\]
	where $\mathit{Targets}_n = \{q\bv \in \conf(\A) \mid \bv(c)\geq n^{2-\varepsilon} \mbox{ for every counter $c$ such that $\by(c)=0$}\}$.
%
\end{lemma}
A proof is in Appendix~\ref{app-lem:-quadratic-lower-bound-proof}.


It remains to prove Theorem~\ref{thm-dichotomy} for VASS MDPs with DAG-like MEC decomposition.
Here, we proceed by analyzing the individual MECs one by one, transferring the output of the previous MEC to the next one. We start in a top MEC with all counters initialized to~$n$. Here we can directly apply Theorem~\ref{lemma-C-limPlinear-if-no-+} to determine which of the $\calC[c]$ have a tight estimate~$n$ and a lower estimate $n^2$, respectively. It follows from Lemma~\ref{lemma-quadratic-lower-bound} that all counters $c$ such that $n^2$ is a lover estimate of $\calC[c]$ can be simultaneously pumped to $n^{2-\varepsilon}$ with very high probability. However, this computation may decrease the counters $c$ such that $n$ is a tight estimate for $\calC[c]$. To ensure that the value of these counters is still $\Omega(n)$ when entering the next MEC, we first divide the initial counter vector $\bn$ into two halves, each of size $\lfloor\frac{\bn}{2}\rfloor$, and then pump the counters $c$ such that $n^2$ is a lower estimate for $\calC[c]$ to the value $(\lfloor\frac{n}{2}\rfloor)^{2-\varepsilon}$. We show that the length of this computation is at most quadratic. The value of the other counters stays at least $\lfloor\frac{n}{2}\rfloor$. When analyzing the next MEC, we treat the counters previously pumped to quadratic values as ``infinite'' because they are sufficiently large so that they cannot prevent pumping additional counters to asymptotically quadratic values. Technically, this is implemented by modifying every counter update vector $\bu$ so that $\bu[c] = 0$ for every ``quadratic'' counter $c$. A precise formulation of these observations and the corresponding proofs are given in Appendix~\ref{app-DAG-like}. 

We conjecture that the dichotomy of Theorem~\ref{thm-dichotomy} holds for \emph{all} VASS MDPs, but we do not have a complete proof. If the MEC decomposition is not DAG-like, a careful analysis of computations revisiting the same MECs is required; such repeated visits may but do not have to enable additional asymptotic growth of~$\calC[c]$.

\section{One-Dimensional VASS MDPs}
\label{sec-one-dim}

In this section, we give a full and effective classification of tight estimates of $\calL$, $\calC[c]$, and $\calT[t]$ for one-dimensional VASS MDPs. More precisely, we prove the following theorem:

\begin{theorem}
	\label{thm-onedim}
	Let $\A$ be a one-dimensional VASS MDP. We have the following:
	\begin{itemize}
		\item Let $c$ be the only counter of $\A$. Then one of the following possibilities holds:
		\begin{itemize}
			\item There exists a type $\beta=M$ such that $\calC[c]$ is unbounded for $\beta$.
			\item $n$ is a tight estimate of $\calC[c]$ for every type.
		\end{itemize}
		\item Let $t$ be a transition of $\A$. Then one of the following possibilities holds:
		\begin{itemize}
			\item There exists a type $\beta=M$ such that $\calT[t]$ is unbounded for $\beta$.
			\item There exists a type $\beta$ such that $\weigth(\beta)>0$ and $\calT[t]$ is unbounded for $\beta$.
			\item There exists a type $\beta=M$ such that $n^2$ is a tight estimate of $\calT[t]$ for $\beta$.
			\item The transition $t$ occurs in some MEC $M$, $n$ is a tight estimate of $\calT[t]$ for every type $\beta$ containing the MEC~$M$, and $0$ is a tight estimate of $\calT[t]$ for every type $\beta$ not containing the MEC~$M$.
			\item The transition $t$ does not occur in any MEC, and for every type $\beta$ of length $k$ we have that $k$ is an upper estimate of $\calT[t]$ for $\beta$.
		\end{itemize}
		\item One of the following possibilities holds:
		\begin{itemize}
			\item There exists a type $\beta=M$ such that $\calL$ is unbounded for $\beta$.
			\item There exists a type $\beta=M$ such that $n^2$ is a tight estimate of $\calL$ for $\beta$.
			\item $n$ is a tight estimate of $\calL$ for every type. 
		\end{itemize}
		
	\end{itemize}
	It is decidable in polynomial time which of the above cases hold.
\end{theorem}
Note that some cases are mutually exclusive and some may hold simultaneously. Also recall that $\weigth(\beta)=1$ for every type $\beta$ of length one, and  $\weigth(\beta)$ decays exponentially in the length of $\beta$. Hence, if a transition $t$ does not occur in any MEC, there is a constant $\kappa< 1$ depending only on $\A$ such that $\prob_{p\bv}^\sigma[\calT[t] \geq i] \leq \kappa^i$ for every $\sigma \in \Sigma$ and $p\bv \in \conf(\A)$.

For the rest of this section, we fix a one-dimensional VASS MDP $\A = \ce{Q, (Q_n,Q_p),T,P}$ and some linear ordering $\sqsubseteq$ on~$Q$.
A proof of Theorem~\ref{thm-onedim} is obtained by analyzing bottom strongly connected components (BSCCs) in a Markov chain obtained from $\A$ by ``applying'' some MD strategy~$\sigma$ (we use $\Sigma_{\MD}$ to denote the class of all MD strategies for $\A$). Recall that $\sigma$ selects the same outgoing transition in every $p \in Q_n$ whenever $p$ is revisited, and hence we can ``apply'' $\sigma$ to $\A$ by removing the other outgoing transitions. The resulting Markov chain is denoted by $\A_\sigma$. Note that every BSCC $\B$ of $\A_\sigma$ can also be seen as an end component of $\A$. For a MEC $M$ of $\A$, we write $\B \subseteq M$ if all states and transitions of $\B$ are included in~$M$.

For every BSCC $\B$ of $\A_\sigma$, let $p_{\B}$ be the least state of $\B$ with respect to $\sqsubseteq$. Let $\update_\mathbb{B}$ be a function assigning to every infinite path $\pi = p_0,\bu_1,p_1,\bu_2,\ldots$ the sum $\sum_{i=1}^{\ell} \bu_i$ if $p_0 = p_{\B}$ and $\ell\geq 1$ is the least index such that $p_\ell = p_{\B}$, otherwise $\update_\mathbb{B}(\pi) = 0$. Hence, $\update_\mathbb{B}(\pi)$ is the change of the (only) counter~$c$ along $\pi$ until $p_{\B}$ is revisited.



\begin{definition}
Let $\B$ be a BSCC of $\A_\sigma$. We say that $\B$ is
\begin{itemize}
	\item \emph{increasing} if $\Exp^\sigma_{p_\mathbb{B}}(\update_{B})>0$,
	\item \emph{decreasing} if $\Exp^\sigma_{p_\mathbb{B}}(\update_{B})<0$,
	\item \emph{bounded-zero} if $\Exp^\sigma_{p_\mathbb{B}}(\update_{B})=0$ and $\prob_{p_\mathbb{B}}^\sigma[\update_{\B}{=}0] =1$,
	\item \emph{unbounded-zero} if $\Exp^\sigma_{p_\mathbb{B}}(\update_{B})=0$ and $\prob_{p_\mathbb{B}}^\sigma[\update_{\B}{=}0] <1$.
\end{itemize}
\end{definition}
Note that the above definition does not depend on the concrete choice of $\sqsubseteq$. We prove the following results relating the existence of upper/lower estimates of $\calL$, $\calC[c]$, and $\calT[t]$ to the existence of BSCCs with certain properties. More concretely,
\begin{itemize}
	\item for $\calC[c]$, we show that
	\begin{itemize}
		\item $\calC[c]$ is unbounded for some type $\beta=M$ if there exists an increasing BSCC $\B$ of $\A_\sigma$ for some $\sigma \in \Sigma_{\MD}$ such that $\B \subseteq M$  (Lemma~\ref{lemma-+-exists});
		\item otherwise, $n$ is a tight estimate of $\calC[c]$ for every type (Lemma~\ref{lemma-no-+-effect-on-C}) 
	\end{itemize}
	\item for $\calL$, we show that
	\begin{itemize}
		\item $\calL$ is unbounded for some type $\beta=M$ if there exists an increasing or bounded-zero BSCC $\B$ of $\A_\sigma$ for some $\sigma \in \Sigma_{\MD}$ such that $\B \subseteq M$ (Lemma~\ref{lemma-+-exists}, Lemma~\ref{lemma-existence-of-0d});
		\item otherwise, $n^2$ is an upper estimate of $\calL$ for every type $\beta$ (Lemma~\ref{lemma-no-+-or-0d-effect-on-L});
		\item if there exists an unbounded-zero BSCC $\B$ of $\A_\sigma$ for some $\sigma \in \Sigma_{\MD}$, then $n^2$ is a lower estimate of $\calL$ for $\beta=M$ where $\B \subseteq M$ (Lemma~\ref{lemma-existence-of-0r});
		\item if every BSCC $\B$ of every $\A_\sigma$ is decreasing, then $\calL_{\eexp}(n)\in \Theta(n)$ (this follows from \cite{BCKNV:probVASS-linear-termination}), and hence $n$ is a tight estimate of $\calL$ for every type (Observation~\ref{observation-upper-expected-implies-upper-estimate});		
	\end{itemize}
	\item for $\calT[t]$, we distinguish two cases:
	\begin{itemize}
		\item If $t$ is not contained in any MEC of $\A$, then for every type $\beta$ of length~$k$, the transition~$t$ cannot be executed more than $k$ times along a arbitrary computation $\pi$ where $\mecs(\pi) = \beta$.
		\item If $t$ is contained in a MEC $M$ of $\A$, then
		\begin{itemize}
			\item $\calT[t]$ is unbounded for $\beta=M$ if there exist an increasing BSCC $\B$ of $\A_\sigma$ for some $\sigma \in \Sigma_{\MD}$ such that $\B \subseteq M$ (Lemma~\ref{lemma-+-exists}), or bounded-zero BSCC $\B$ of $\A_\sigma$ for some $\sigma \in \Sigma_{\MD}$ such that $\B$ contains~$t$ (Lemma~\ref{lemma-existence-of-0d});
			\item $\calT[t]$ is unbounded for every $\beta=M_1,\dots, M_k$ such that $M = M_i$  for some $i$ and there exists an increasing BSCC $\B$ of $\A_\sigma$ for some $\sigma \in \Sigma_{\MD}$ such that $\B \subseteq M_j$ for some $j \leq i$ (Lemma~\ref{lemma-+-exists});
			\item otherwise, $n^2$ is an upper estimate of $\calT[t]$ for every type  (Lemma~\ref{lemma-no-+-or-0d-effect-on-L});
			\item if there is an unbounded-zero BSCC $\B$ of $\A_\sigma$ for some $\sigma \in \Sigma_{\MD}$ such that $\B$ contains~$t$, then $n^2$ is a lower estimate of $\calT[t]$ for $\beta=M$ (Lemma~\ref{lemma-existence-of-0r});
			\item if every BSCC $\B$ of every $\A_\sigma$ is decreasing, then $\calT[t]_{\eexp}(n)\in \Theta(n)$ (this follows from \cite{BCKNV:probVASS-linear-termination}), and hence $n$ is an upper estimate of $\calT[t]$ for every type (Observation~\ref{observation-upper-expected-implies-upper-estimate}).
		\end{itemize}
	\end{itemize}
\end{itemize}

The polynomial time bound of Theorem~\ref{thm-onedim} is then obtained by realizing the following: First, we need to decide the existence of an increasing BSCC of $\A_\sigma$ for some $\sigma \in \Sigma_{\MD}$. This can be done in polynomial time using the constraint system~(I) of Fig.~\ref{fig-systems} (Lemma~\ref{decide-+-in-P}). If no such increasing BSCC exists, we need to decide the existence of a bounded-zero BSCC, which can be achieved in polynomial time for a subclass of one-dimensional VASS MDPs where no increasing BSCC exists (Lemma~\ref{lemma-0d-decidable-in-P-when-no-+}). Then, if no bounded-zero BSCC exists, we need to decide the existence of an unbounded-zero BSCC, which can again be done in polynomial time using the constraint system~(I) of Fig.~\ref{fig-systems} (realize that any solution $\bx$ of (I) implies the existence of a BSCC that is either increasing, bounded-zero, or unbounded-zero).

Hence, the ``algorithmic part'' of Theorem~\ref{thm-onedim} is an easy consequence of the above observations, but there is one remarkable subtlety. Note that we need to decide the existence of a bounded-zero BSCC only for a subclass of one-dimensional VASS MDPs where no increasing BSCCs exist. This is actually crucial, because deciding the existence of a bounded-zero BSCC in \emph{general} one-dimensional VASS MDPs is \textbf{NP}-complete (Lemma~\ref{lemma-np-hard-0d-gen}). 

The main difficulties requiring novel insights are related to proving the observation about $\C[c]$, stating that if there is no increasing BSCC of $\A_\sigma$ for any $\sigma \in \Sigma_{\MD}$, then $n$ is an upper estimate of $\C[c]$ for every type. A comparably difficult (and in fact closely related) task is to show that if there is no increasing or bounded-zero BSCC, then $n^2$ is an upper estimate of $\calL$ for every type. Note that here we need to analyze the behaviour of $\A$ under \emph{all} strategies (not just MD), and consider the notoriously difficult case when the long-run average change of the counter caused by applying the strategy is zero. Here we need to devise a suitable decomposition technique allowing for interpreting general strategies as ``interleavings'' of MD strategies and lifting the properties of MD strategies to general strategies. Furthermore, we need to devise techniques for reducing the problems of our interest to analyzing certain types of random walks that have already been studied in stochastic process theory. We discuss this more in the following subsection, and we refer to Appendix~\ref{app-one-dim} for a complete exposition of these results.

\subsection{MD decomposition}


As we already noted, one crucial observation behind Theorem~\ref{thm-onedim} is that if there is no increasing BSCC of $\A_\sigma$ for any $\sigma \in \Sigma_{\MD}$, then $n$ is an upper estimate of $\C[c]$ for every type. In this section, we sketch the main steps towards this result.

First, we show that every path in $\A$ can be decomposed into ``interweavings'' of paths generated by MD strategies. 

Let \(\alpha=p_0,\bv_1,\dots,p_k \) be a path. For every $i \leq k$, we use \(\alpha_{..i}=p_0,\bv_1, \ldots,p_i \) to denote the prefix of \(\alpha \) of length \(i\). We say that $\alpha$ is \emph{compatible} with a MD strategy $\sigma$ if $\sigma(\alpha_{..i}) = (p_i,\bv_{i+1},p_{i+1})$ for all $i<k$ such that $p_{i} \in Q_n$. Furthermore, for every path \(\beta=q_0,\bu_1,q_1,\dots,q_\ell \) such that $p_k=q_0$, we define a path \(\alpha \circ \beta= p_0,\bv_1,p_1,\dots,p_k,\bu_1,q_1,\dots,q_\ell \).

\begin{definition}
	Let $\A$ be a VASS MDP, $\pi_1,\dots, \pi_k \in \Sigma_{\MD}$, and $p_1,\dots,p_k \in Q$. An \mbox{\emph{MD-decomposition}} of a path $\alpha = s_1,\ldots,s_m$ under  
	$\pi_1,\dots, \pi_k$ and $p_1,\dots,p_k$ is a decomposition of $\alpha$ into finitely many paths $\alpha = \gamma_1^1 \circ \cdots\circ\gamma_1^k \; \circ\; \gamma_2^1\circ\cdots\circ\gamma_2^k \; \circ \; \cdots \; \circ \; \gamma_\ell^1\circ\cdots\circ\gamma_\ell^k$ satisfying the following conditions:
	\begin{itemize}
		\item for all $i < \ell$ and $j\leq k$, the last state of $\gamma_i^j$ is the same as the first state of $\gamma_{i+1}^j$;
		\item for every \(j\leq k \), $\gamma_1^j \circ \cdots \circ \gamma_\ell^j$ is a path that begins with \(p_j \) and is compatible with $\pi_j$.   
	\end{itemize} 
\end{definition}
Note that $\pi_1,\dots, \pi_k$ and $p_1,\dots,p_k$ are not necessarily pairwise different, and the length of $\gamma_{i}^j$ can be zero. Also note that the same $\alpha$ may have several MD-decompositions. 

Intuitively, an MD decomposition of $\alpha$ shows how to obtain $\alpha$ by repeatedly selecting zero or more transitions by $\pi_1,\ldots,\pi_k$. The next lemma shows that for every VASS MDP $\A$, one can \emph{fix} MD strategies $\pi_1,\ldots,\pi_k$ and states $p_1,\ldots,p_k$ such that \emph{every} path $\alpha$ in $\A$ has an MD-decomposition under $\pi_1,\ldots,\pi_k$ and $p_1,\ldots,p_k$. Furthermore, such a decomposition is constructible \emph{online} as $\alpha$ is read from left to right.

\begin{lemma}\label{lemma-MD-decompose}
	For every VASS MDP \(\A\), there exist \(\pi_1,\dots,\pi_k \in \Sigma_{\MD}\),  \(p_1,\dots,p_k \in Q\), and a function \(\Decompose_\A \) such that the following conditions are satisfied for every finite path $\alpha$:
	\begin{itemize}
		\item \(\Decompose_\A(\alpha) \) returns an MD-decomposition of $\alpha$ under $\pi_1,\dots,\pi_k$ and $p_1,\ldots,p_k$.
		\item $\Decompose_\A(\alpha)=\Decompose_\A(\alpha_{..\length(\alpha)-1})\; \circ\; \gamma^1\circ\cdots\circ\gamma^k$, where exactly one of 
		$\gamma^i$ has positive length (the $i$ is called the \emph{mode} of $\alpha$).
		\item If the last state of \(\alpha_{..\length(\alpha)-1}\) is probabilistic, then the mode of $\alpha$ does not depend on the last transition of $\alpha$. 
	\end{itemize}
\end{lemma}	
A proof of Lemma~\ref{lemma-MD-decompose} is in Appendix~\ref{app-one-dim}.

According to Lemma~\ref{lemma-MD-decompose}, \emph{every} strategy $\sigma$ for $\A$ just performs a certain ``interleaving'' of the MD strategies  $\pi_1,\dots,\pi_k$ initiated in the states $p_1,\ldots,p_k$. We aim to show that if every BSCC of every $\A_{\pi_j}$ is non-increasing, then $n$ is an upper estimate of $\C[c]$ for every type. Since we do not have any control over the length of the individual $\gamma_{i}^j$ occurring in MD-decompositions, we need to introduce another concept of \emph{extended VASS MDPs} where the strategies $\pi_1,\dots,\pi_k$ can be interleaved in ``longer chunks''. Intuitively, an extended VASS MDP is obtained from $\A$ by taking $k$ copies of $\A$ sharing the same counter. The $j$-th copy selects transitions according to $\pi_j$. At each round, only one $\pi_j$ makes a move, where the $j$ is selected by a special type of ``pointing'' strategy defined especially for extended MDPs.
Note that $\sigma$ can be faithfully simulated in the extended VASS MDP by a pointing strategy that selects the indexes consistently with $\Decompose_\A$. However, we can also construct another pointing strategy that simulates each $\pi_j$ longer (i.e., ``precomputes'' the steps executed by $\pi_j$ in the future) and thus ``close cycles'' in the BSCC visited by $\pi_j$.  This computation can be seen as an interleaving of a finite number of independent random walks with non-positive expectations. Then, we use the optional stopping theorem to get an upper bound on the total expected number of ``cycles'', which can then be used to obtain the desired upper estimate. We refer to Appendix~\ref{app-one-dim} for details.

\subsection{A Note about Energy Games}

One-dimensional VASS MDPs are closely related to energy games/MDPs \cite{CHD:energy-games,CHD:energy-MDPs,CHKN:energy-games-polynomial,JLS:eVASS-games-ICALP}. An important open problem for energy games is the complexity of deciding the existence of a \emph{safe} configuration where, for a sufficiently high energy amount, the responsible player can avoid decreasing the energy resource (counter) below zero. This problem is known to be in $\NP\cap \coNP$, and a pseudopolynomial algorithm for the problem exists; however, it is still open whether the problem is in $\PTIME$ when the counter updates are encoded in binary. Our analysis shows that this problem is solvable in polynomial time for energy (i.e., one-dimensional VASS) MDPs $\A$ such that there is no increasing SCC of $\A_\sigma$ for any $\sigma \in \Sigma_{\MD}$.


We say that a SCC $\B$ of $\A_\sigma$ is \emph{non-decreasing} if $\B$ does not contain any negative cycles. Note that every bounded-zero SCC is non-decreasing, and a increasing SCC may but does not have to be non-decreasing. 

\begin{lemma}
	\label{lem-energy}
	An energy MDP has a safe configuration iff there exists a non-decreasing SCC $\B$ of $\A_\sigma$ for some $\sigma \in \Sigma_{\MD}$. 
\end{lemma}
The ``$\Leftarrow$'' direction of Lemma~\ref{lem-energy} is immediate, and the other direction can be proven using our MD decomposition technique, see Appendix~\ref{app-lem-energy-proof}. 

Note that if there is no increasing SCC $\B$ of $\A_{\sigma}$ for any $\sigma \in \Sigma_{\MD}$, then the existence of a non-decreasing SCC is equivalent to the existence of a bounded-zero SCC, and hence it can be decided in polynomial time (see the results presented above). However, for general energy MDPs, the best upper complexity bound for the existence of a non-decreasing SCC is $\NP\cap \coNP$. Interestingly, a small modification of this problem already leads to $\NP$-completeness, as demonstrated by the following lemma.

\begin{lemma}\label{lemma-type-NONNEGATIVE-NP-complete}
    The problem whether there exists a non-decreasing SCC $\B$ of $\A_{\sigma}$ for some $\sigma \in \Sigma_{\MD}$ such that $\B$ contains a given state $p \in Q$ is $\NP$-complete.
\end{lemma}
A proof of Lemma~\ref{lemma-type-NONNEGATIVE-NP-complete} is in Appendix~\ref{app-lemma-type-NONNEGATIVE-NP-complete-proof}.

\section{Conclusions}
\label{sec-concl}

We introduced new estimates for measuring the asymptotic complexity of probabilistic programs and their VASS abstractions. We demonstrated the advantages of these measures over the asymptotic analysis of expected values, and we have also shown that tight complexity estimates can be computed efficiently for certain subclasses of VASS MDPs.

A natural continuation of our work is extending the results achieved for one-dimensional VASS MDPs to the multi-dimensional case. In particular, an interesting open question is whether the polynomial asymptotic analysis for non-deterministic VASS presented in \cite{Zuleger:VASS-polynomial} can be generalized to VASS MDPs. Since the study of multi-dimensional VASS MDPs is notoriously difficult, a good starting point would be a complete understanding of VASS MDPs with two counters. 





\bibliography{str-short,concur}

\begin{thebibliography}{10}

\bibitem{AK:VASS-polynomial-termination}
M.~Ajdar{\'{o}}w and A.~Ku{\v{c}}era.
\newblock Deciding polynomial termination complexity for {VASS} programs.
\newblock In {\em Proceedings of CONCUR 2021}, volume 203 of {\em LIPIcs},
  pages 30:1--30:15. Schloss Dagstuhl, 2021.

\bibitem{BCKNV:probVASS-linear-termination}
T.~Br\'{a}zdil, K.~Chatterjee, A.~Ku{\v{c}}era, P.~Novotn{\'{y}}, and D.~Velan.
\newblock Deciding fast termination for probabilistic {VASS} with
  nondeterminism.
\newblock In {\em Proceedings of {ATVA 2019}}, volume 11781 of {\em LNCS},
  pages 462--478. Springer, 2019.

\bibitem{BCKNVZ:VASS-linear-termination}
T.~Br\'{a}zdil, K.~Chatterjee, A.~Ku{\v{c}}era, P.~Novotn{\'{y}}, D.~Velan, and
  F.~Zuleger.
\newblock Efficient algorithms for asymptotic bounds on termination time in
  {VASS}.
\newblock In {\em Proceedings of LICS 2018}, pages 185--194. ACM Press, 2018.

\bibitem{BW:nondet-languages}
M.~Broy and M.~Wirsing.
\newblock On the algebraic specification of nondeterministic programming
  languages.
\newblock In {\em Proceedings of CAAP'81}, volume 112 of {\em LNCS}, pages
  162--179. Springer, 1981.

\bibitem{CHD:energy-games}
K.~Chatterjee and L.~Doyen.
\newblock Energy parity games.
\newblock In {\em Proceedings of ICALP 2010, Part II}, volume 6199 of {\em
  LNCS}, pages 599--610. Springer, 2010.

\bibitem{CHD:energy-MDPs}
K.~Chatterjee and L.~Doyen.
\newblock Energy and mean-payoff parity {Markov} decision processes.
\newblock In {\em Proceedings of MFCS 2011}, volume 6907 of {\em LNCS}, pages
  206--218. Springer, 2011.

\bibitem{CHKN:energy-games-polynomial}
K.~Chatterjee, M.~Henzinger, S.~Krinninger, and D.~Nanongkai.
\newblock Polynomial-time algorithms for energy games with special weight
  structures.
\newblock In {\em Proceedings of ESA 2012}, volume 7501 of {\em LNCS}, pages
  301--312. Springer, 2012.

\bibitem{CLLLM:VASS-reach-nonelem}
W.~Czerwi{\'{n}}ski, S.~Lasota, R.~Lazi{\'{c}}, J.~Leroux, and F.~Mazowiecki.
\newblock The reachability problem for {Petri} nets is not elementary.
\newblock In {\em Proceedings of STOC 2019}, pages 24--33. ACM Press, 2019.

\bibitem{Esparza:ModelChecking-AI}
J.~Esparza.
\newblock Decidability of model checking for infinite-state concurrent systems.
\newblock {\em AI}, 34:85--107, 1997.

\bibitem{HP:VASS-reachability-TCS}
J.E.{} Hopcroft and J.-J.{} Pansiot.
\newblock On the reachability problem for 5-dimensional vector addition
  systems.
\newblock {\em TCS}, 8:135--159, 1979.

\bibitem{Jancar:PN-bisimilarity-TCS}
P.~Jan{\v{c}}ar.
\newblock Undecidability of bisimilarity for {P}etri nets and some related
  problems.
\newblock {\em TCS}, 148(2):281--301, 1995.

\bibitem{JLS:eVASS-games-ICALP}
M.~Jurdzi{\'{n}}ski, R.~Lazi{\'{c}}, and S.~Schmitz.
\newblock Fixed-dimensional energy games are in pseudo-polynomial time.
\newblock In {\em Proceedings of ICALP 2015}, volume 9135 of {\em LNCS}, pages
  260--272. Springer, 2015.

\bibitem{Kucera:Asymptotic-VASS-Analysis-SIGLOG}
A.~Ku{\v{c}}era.
\newblock Algorithmic analysis of termination and counter complexity in vector
  addition systems with states: A survey of recent results.
\newblock {\em ACM SIGLOG News}, 8(4):4--21, 2021.

\bibitem{KLV:VASS-Grzegorczyk}
A.~Ku{\v{c}}era, J.~Leroux, and D.~Velan.
\newblock Efficient analysis of {VASS} termination complexity.
\newblock In {\em Proceedings of LICS 2020}, pages 676--688. ACM Press, 2020.

\bibitem{Leroux:Polynomial-termination-VASS}
J.~Leroux.
\newblock Polynomial vector addition systems with states.
\newblock In {\em Proceedings of ICALP 2018}, volume 107 of {\em LIPIcs}, pages
  134:1--134:13. Schloss Dagstuhl, 2018.

\bibitem{LS:Petri-computer}
J.~Leroux and Ph.{} Schnoebelen.
\newblock On functions weakly computable by {Petri} nets and vector addition
  systems.
\newblock In {\em Reachability Problems}, volume 8762 of {\em LNCS}, pages
  190--202. Springer, 2014.

\bibitem{Lipton:PN-Reachability}
R.~Lipton.
\newblock The reachability problem requires exponential space.
\newblock Technical report~62, Yale University, 1976.

\bibitem{MM:containment-Petri}
E.W.{} Mayr and A.R.{} Meyer.
\newblock The complexity of the finite containment problem for {Petri} nets.
\newblock {\em JACM}, 28(3):561--576, 1981.

\bibitem{Petri:first-paper}
C.A. Petri.
\newblock Kommunikation mit automaten.
\newblock {\em Schriften des Institutes f{\"{u}}r Instrumentelle Mathematik},
  3, 1962.

\bibitem{Puterman:book}
M.L.{} Puterman.
\newblock {\em {Markov} Decision Processes}.
\newblock Wiley, 1994.

\bibitem{SZV:amortized}
M.~Sinn, F.~Zuleger, and H.~Veith.
\newblock A simple and scalable static analysis for bound analysis and
  amortized complexity analysis.
\newblock In {\em Proceedings of CAV 2014}, volume 8559 of {\em LNCS}, pages
  745--761. Springer, 2013.

\bibitem{SZV:difference-constraints}
M.~Sinn, F.~Zuleger, and H.~Veith.
\newblock Complexity and resource bound analysis of imperative programs using
  difference constraints.
\newblock {\em Journal of Automated Reasoning}, 59(1):3--45, 2017.

\bibitem{Zuleger:VASS-polynomial}
F.~Zuleger.
\newblock The polynomial complexity of vector addition systems with states.
\newblock In {\em Proceedings of {FoSSaCS 2020}}, volume 12077 of {\em LNCS},
  pages 622--641. Springer, 2020.

\end{thebibliography}

\appendix
\newpage

\section{Proofs for Section~\ref{section-linquad}}
\label{app-linquad}
\subsection{Proof of Lemma~\ref{lemma:dichotomy}}\label{app-lem:optimization-duality-proof}

\begin{lemma*}[\textbf{\ref{lemma:dichotomy}}]
Let $\bx$ be a (maximal) solution to the constraint system~(I) and $\by,\bz$ be a (maximal) solution to the constraint system~(II). Then, for each counter $c$ we have that either $\by(c)>0$ or $\sum_{t \in T} \bx(t)\bu_t(c)>0$, and for each transition $t = (p,\bu,q)\in T$ we have that
\begin{itemize}
	\item if $p\in Q_n$ then either $\bz(q)-\bz(p)+\sum_{i=1}^{d} \bu(i)\by(i)<0$ or $\bx(t)>0$;
	\item if $p\in Q_p$ then either \[\sum_{t'=(p,\bu',q') \in \tout(p)}P(t')\big(\bz(q')-\bz(p)+\sum_{i=1}^{d}  \bu'(i)\by(i)\big)< 0\] or $\bx(t)>0$.		
\end{itemize}


\end{lemma*}


\noindent
\textbf{Proof:} 
Let \(A=T_n\cup Q_p\), where \(T_n=\bigcup_{p\in Q_n}\tout(p) \). For each \(a\in A\), let $\next_a$ be a probability distribution on \(Q\) such that
\begin{itemize}
	\item \(\next_{(p,\bu,q)}(q)=1 \) for \(a=(p,\bu,q)\in T_n \),
	\item \(\next_p(q)=\sum_{(p,\bu,q)\in \tout(p)\cap \tin(q)} P((p,\bu,q)) \) for \(a=p\in Q_p \),
	\item and \(\next_a(p)=0 \) else,
\end{itemize}   
let \(\from_a\) be a probability distribution on \(Q\) such that
\begin{itemize}
	\item \(\from_{(p,\bu,q)}(p)=1 \) for \(a=(p,\bu,q)\in T_n \),
	\item \(\from_p(p)=1 \) for \(a=p\in Q_p \),
	\item and \(\from_a(p)=0 \) else,
\end{itemize}   
 and let \(\effect_a\) be defined as \begin{itemize}
 	\item \(\effect_{(p,\bu,q)}=\bu \) for \(a=(p,\bu,q)\in T_n \),
 	\item  	and \(\effect_p=\sum_{(p,\bu,q)\in \tout(p)} P((p,\bu,q))\bu \) for \(a=p\in Q_p \). 
 \end{itemize}  

 \noindent
Then we can rewrite the constraint systems as
\smallskip

\noindent
	\vspace{0.7cm}
	\begin{tabular}{|c|c|}
		\hline
		{\begin{minipage}[c]{0.35\textwidth}\small
				\vspace{0.2cm}
				Constraint system~(I'):
				
				\vspace{0.2cm}
				Find $\bx' \in \mathbb{Z}^{A}$ such that
				\begin{align}
					\sum_{a \in A} \bx'(a) \effect_a  & \ge \vec{0} \nonumber\\
					\bx' & \ge \vec{0} \nonumber
				\end{align}
				
				and for each $p\in Q$
				\begin{align}
					\sum_{a\in A}(\bx'(a)\next_a(p)-\bx'(a)\from_a(p))&=0\nonumber
				\end{align}
				%
				%
				%
%
				\vspace*{1em}
		\end{minipage}}
		&
		{\begin{minipage}[c]{0.6\textwidth}
				\vspace{0.2cm}
				Constraint system (II'):
				
				\vspace{0.2cm}
				Find $\by \in \mathbb{Z}^d,\bz \in \mathbb{Z}^{Q}$ such that
				\begin{align}
					\by & \ge  \vec{0} \nonumber\\
					\bz & \ge  \vec{0} \nonumber
				\end{align}
				and for each $a\in A$  				\[
					\sum_{p\in Q} (\bz(p)\next_a(p) - \bz(p)(\from_a(p))) + \sum_{i=1}^d \effect_a(i)\by(i)\leq 0
				\]
%
				\vspace*{1em}
		\end{minipage}}\\
		\hline
	\end{tabular}

We recognize systems (I) and (I') as equivalent, and systems (II) and (II') as equivalent as per the following lemma.

\begin{lemma}
If  \(\bx',\by,\bz\) is a solution to the rewritten constraint systems (I') and (II'), then \(\bx,\by,\bz \) is a solution to the original constraint systems (I) and (II), where \(\bx(t)=\bx'(t) \) for \(t\in T_n \), and \(\bx((p,\bu,q))=P((p,\bu,q))\bx'(p) \) for \((p,\bu,q)\in T\setminus T_n \). Similarly, if \(\bx,\by,\bz \) is a solution to the original constraint systems (I) and (II), then \(\bx',\by,\bz\) is a solution to the rewritten constraint systems (I') and (II'), where \(\bx'(t)=\bx(t) \) for \(t\in T_n \), and \(\bx'(p)=\sum_{t\in \tout(p)} \bx(t) \) for \(p\in Q_p\).
	\end{lemma}

\begin{proof}
	\textbf{The first half (I):} Let \(\bx' \) be a solution of (I'), we will show that \(\bx\) is a solution to (I), where \(\bx(t)=\bx'(t) \) for \(t\in T_n \), and \(\bx((p,\bu,q))=P((p,\bu,q))\bx'(p) \) for \((p,\bu,q)\in T\setminus T_n \).
	
	It holds from (I') that
	 \begin{align*}
		&\sum_{a \in A} \bx'(a) \effect_a 
		=
		 \sum_{t \in T_n} \bx'(t) \effect_t + \sum_{p \in Q_p} \bx'(p) \effect_p 
		 =\\&=
		  \sum_{t \in T_n} \bx'(t) \bu_t + \sum_{p \in Q_p} \bx'(p) (\sum_{(p,\bu,q)\in \tout(p)} P((p,\bu,q))\bu)
		  \\&=
		  \sum_{t \in T_n} \bx(t) \bu_t + \sum_{p \in Q_p} \sum_{(p,\bu,q)\in \tout(p)}\bx'(p) P((p,\bu,q))\bu
		  \\&=
		  \sum_{t \in T_n} \bx(t) \bu_t + \sum_{p \in Q_p} \sum_{(p,\bu,q)\in \tout(p)}\bx((p,\bu,q))\bu
		  \\&=
		  \sum_{t \in T} \bx(t) \bu_t \geq \vec{0}
	\end{align*}

	\(\bx\geq 0 \) holds from both \(\bx'\geq 0 \) and \(P(t)\geq 0 \) for each \(t\in T\setminus T_n\).
	
	For each \(p\in Q\) it holds from (I')
	
	\begin{align*}
		&\sum_{a\in A}(\bx'(a)\next_a(p)-\bx'(a)\from_a(p))
		=\\&=
		 \sum_{a\in T_n}(\bx'(a)\next_a(p)-\bx'(a)\from_a(p)) +
		 \sum_{a\in Q_r}(\bx'(a)\next_a(p)-\bx'(a)\from_a(p))
		 \\&=
		 \sum_{t\in \tin(p)\cap T_n}\bx'(t) - \sum_{t\in \tout(p)\cap T_n}\bx'(t) +
		 \sum_{a\in Q_r}\bx'(a)(\sum_{t\in \tout(a)\cap \tin(p)} P(t))- \bx'(p)
		 \\&=
		 \sum_{t\in \tin(p)\cap T_n}\bx(t) - \sum_{t\in \tout(p)\cap T_n}\bx(t) +
		 \sum_{a\in Q_r}\sum_{t\in \tout(a)\cap \tin(p)} \bx'(a)P(t)- \sum_{t\in \tout(p) }P(t)\bx'(p)
		 \\&=
		 \sum_{t\in \tin(p)\cap T_n}\bx(t) - \sum_{t\in \tout(p)\cap T_n}\bx(t) +
		\sum_{a\in Q_r}\sum_{t\in \tout(a)\cap \tin(p)} \bx(t)- \sum_{t\in \tout(p) }\bx(t)	
		\\&=
		 \sum_{t\in \tin(p)}\bx(t) - \sum_{t\in \tout(p)}\bx(t) =0			 
	\end{align*}
	
	And for each \(p\in Q_p, t\in \tout(p) \) it holds  \(\sum_{t'\in\tout(p)}\bx(t')=\sum_{t'\in\tout(p)}P(t')\bx'(p)=\bx'(p) \), therefore it holds \(\bx(t)=P(t)\bx'(p)= P(t)\sum_{t'\in\tout(p)}x(t') \).
	
	Thus \(\bx \) is a solution to (I).
	
	\textbf{The first half (II):} Let \(\by,\bz \) be a solution of (II'), we will show it is also a solution of (II).
	
	For each \(a=(p,\bu,q)\in T_n\) it holds from (II') 
	\begin{align*}
		\sum_{p'\in Q} (\bz(p')(\next_a(p') - \bz(p')(\from_a(p'))) + \sum_{i=1}^d \effect_a(i)\by(i)
		=
		\bz(q) - \bz(p) + \sum_{i=1}^d \bu(i)\by(i)		
		\leq
		 0
	\end{align*}
And for each \(a=p\in Q_p \) it holds from (II')
\begin{align*}
	&\sum_{q\in Q} (\bz(q)\next_a(q) - \bz(q)\from_a(q))) + \sum_{i=1}^d \effect_a(i)\by(i)
	=\\&=
	\sum_{q\in Q} (\bz(q)(\sum_{t\in \tout(p)\cap \tin(q)} P(t))) - \bz(p) + \sum_{i=1}^d \effect_a(i)\by(i)
	\\&=
	\sum_{t\in \tout(p)} \bz(q)P(t) - \sum_{t\in \tout(p)} \bz(p)P(t) + \sum_{i=1}^d \effect_a(i)\by(i)
	\\&=
	\sum_{t\in \tout(p)} P(t)(\bz(q) - \bz(p)) + \sum_{i=1}^d \sum_{t\in \tout(p)} P(t)\bu_t(i)\by(i)
	\\&=
	\sum_{t\in \tout(p)} P(t)(\bz(q) - \bz(p)) + \sum_{t\in \tout(p)}\sum_{i=1}^d  P(t)\bu_t(i)\by(i)
	\\&=
	\sum_{t\in \tout(p)} P(t)\big( \bz(q) - \bz(p) + \sum_{i=1}^d \bu_t(i)\by(i) \big)	
		\leq
	0
\end{align*}
Therefore \(\by,\bz \) is a solution of (II).

	\textbf{The second half (I'):} Let \(\bx\) be a solution of (I), we will show that \(\bx' \) is a solution of (I'), where \(\bx'(t)=\bx(t) \) for \(t\in T_n \), and \(\bx'(p)=\sum_{t\in \tout(p)} \bx(t) \) for \(p\in Q_p\).
	
	From (I) it holds  for \(T_p=\bigcup_{p\in Q_p}\tout(p) \)	\begin{align*}
		&\sum_{t \in T} \bx(t) \bu_t 
		=\\&=
		\sum_{t \in T_n} \bx(t) \bu_t + \sum_{t \in T_p} \bx(t) \bu_t 
		\\&=
		\sum_{t \in T_n} \bx'(t) \bu_t + \sum_{(p,\bu,q) \in T_p} \bu P((p,\bu,q))\cdot \big(\sum_{t'\in \tout(p)} \bx(t')\big) 
		\\&=
		\sum_{t \in T_n} \bx'(t) \bu_t + \sum_{(p,\bu,q) \in T_p} \bu P((p,\bu,q)) \bx'(p) 
		\\&=
		\sum_{t \in T_n} \bx'(t) \bu_t + \sum_{p\in Q_p} \sum_{(p,\bu,q) \in \tout(p)} \bu P((p,\bu,q)) \bx'(p)
		\\&=
		\sum_{t \in T_n} \bx'(t)\cdot \effect_t + \sum_{p\in Q_p} \bx'(p)\cdot \effect_p
		\\&=
		\sum_{a \in A} \bx'(a)\cdot \effect_a
		\geq 
		\vec{0}		 		 
	\end{align*}
	We get \(\bx'\geq 0 \) trivially from (I).

It also holds
\begin{align*}
	&\sum_{a\in A}(\bx'(a)\next_a(p)-\bx'(a)\from_a(p))
	=\\&=
	\sum_{t\in T_n}(\bx'(t)\next_t(p)-\bx'(t)\from_t(p))+
	\sum_{q\in Q_p}(\bx'(q)\next_q(p)-\bx'(q)\from_q(p))
	\\&=
	\sum_{t\in \tin(p)\cap T_n}\bx'(t)
	-\sum_{t\in \tout(p)\cap T_n}\bx'(t)
	+\sum_{q\in Q_p} \bx'(q)(\sum_{t\in \tout(q)\cap \tin(p)} P(t))
	-\sum_{q\in Q_p\cap \{p \}} \bx'(q)
	\\&=
	\sum_{t\in \tin(p)\cap T_n}\bx'(t)
	-\sum_{t\in \tout(p)\cap T_n}\bx'(t)
	+\sum_{q\in Q_p} \sum_{t\in \tout(q)\cap \tin(p)} P(t)\bx'(q)
	-\sum_{q\in Q_p\cap \{p \}} \bx'(q)
	\\&=
	\sum_{t\in \tin(p)\cap T_n}\bx(t)
	-\sum_{t\in \tout(p)\cap T_n}\bx(t)
	+\sum_{q\in Q_p} \sum_{t\in \tout(q)\cap \tin(p)} P(t)\cdot (\sum_{t'\in \tout(q)} \bx(t'))
	-\\&-\sum_{q\in Q_p\cap \{p \}} \sum_{t\in \tout(q)} \bx(t)
	\\&=
	\sum_{t\in \tin(p)\cap T_n}\bx(t)
	-\sum_{t\in \tout(p)\cap T_n}\bx(t)
	+\sum_{q\in Q_p} \sum_{t\in \tout(q)\cap \tin(p)} \bx(t)
	-\sum_{q\in Q_p\cap \{p \}} \sum_{t\in \tout(q)} \bx(t)
	\\&=
	\sum_{t\in \tin(p)\cap T_n}\bx(t)
	-\sum_{t\in \tout(p)\cap T_n}\bx(t)
	+\sum_{t\in T_p\cap \tin(p)} \bx(t)
	-\sum_{q\in Q_p\cap \{p \}} \sum_{t\in \tout(q)} \bx(t)
\end{align*}

If \(p\in Q_n\), then this becomes 
\begin{align*}
	&\sum_{t\in \tin(p)\cap T_n}\bx(t)
	-\sum_{t\in \tout(p)\cap T_n}\bx(t)
	+\sum_{t\in T_p\cap \tin(p)} \bx(t)
	-\sum_{q\in Q_p\cap \{p \}} \sum_{t\in \tout(q)} \bx(t)
	=\\&=
	\sum_{t\in \tin(p)\cap T_n}\bx(t)
	-\sum_{t\in \tout(p)\cap T_n}\bx(t)
	+\sum_{t\in T_p\cap \tin(p)} \bx(t)
	-0
	=\\&=
	\sum_{t\in \tin(p)}\bx(t)
	-\sum_{t\in \tout(p)\cap T_n}\bx(t)
	=
	\sum_{t\in \tin(p)}\bx(t)
	-\sum_{t\in \tout(p)}\bx(t)
	=
	0
\end{align*}
with the last line being from (I). 
And if \(p\in Q_p\), then it becomes 
\begin{align*}
	&\sum_{t\in \tin(p)\cap T_n}\bx(t)
	-0
	+\sum_{t\in T_p\cap \tin(p)} \bx(t)
	-\sum_{q\in Q_p\cap \{p \}} \sum_{t\in \tout(q)} \bx(t)
	=\\&=
	\sum_{t\in \tin(p)\cap T_n}\bx(t)
	+\sum_{t\in T_p\cap \tin(p)} \bx(t)
	-\sum_{t\in \tout(p)} \bx(t)
	=\\&=
	\sum_{t\in \tin(p)}\bx(t)
	-\sum_{t\in \tout(p)} \bx(t)
	=
	0	
\end{align*}
with the last line being from (I). Therefore \(\bx' \) is a solution of (I')
	
%
	
	\textbf{The second half (II'):} Let \(\by,\bz \) be a solution of (II) we will show that \(\by,\bz \) is also a solution of (II').

	From (II) we have for each \(t=(p,\bu,q)\in T_n \) that 
	\begin{align*}
		&\bz(q)-\bz(p)+\sum_{i=1}^{d} \bu(i)\by(i)
		=\\&=
		\bz(q)\cdot \next_t(q)-\bz(p)\cdot \from_t(p)+\sum_{i=1}^{d} \effect_t(i)\by(i)
		=\\&= \sum_{r\in Q }\bz(r)\cdot \next_t(r)-\bz(r)\cdot \from_t(r))+\sum_{i=1}^{d} \effect_t(i)\by(i)
		\leq 0
	\end{align*}
	Where we used that \(\next_t(r)=0 \) for every \(r\neq q \), and \(\from_t(r)=0 \) for every \(r\neq p \).
	
	Additionally, From (II) we also have for each \(p\in Q_p \) that 
	\begin{align*}
		&\sum_{t= (p,\bu,q) \in \tout(p)}P(t)\big(\bz(q)-\bz(p)+\sum_{i=1}^{d}  \bu_t(i)\by(i)\big)
		=\\&=
		-\bz(p)+\sum_{t= (p,\bu,q) \in \tout(p)}P(t)\big(\bz(q)+\sum_{i=1}^{d}  \bu_t(i)\by(i)\big)
		=\\&=
		-\bz(p)+\sum_{q\in Q} \sum_{t\in \tout(p)\cap \tin(q)}P(t)\big(\bz(q)+\sum_{i=1}^{d}  \bu_t(i)\by(i)\big)
		=\\&=
		-\bz(p)+\sum_{q\in Q} \sum_{t \in \tout(p)\cap \tin(q)}P(t)\bz(q)+ \sum_{t\in \tout(p)}P(t) \big( \sum_{i=1}^{d}  \bu_t(i)\by(i)\big)
		=\\&=
		-\bz(p)+\sum_{q\in Q}\bz(q) \sum_{t\in \tout(p)\cap \tin(q)}P(t) + \sum_{i=1}^{d} \by(i)\cdot(\sum_{t \in \tout(p)}P(t) \bu_t(i))
		=\\&=
		-\bz(p)\cdot \from_p(p)+\sum_{q\in Q}\bz(q)\cdot \next_p(q)  + \sum_{i=1}^{d} \by(i)\cdot \effect_p (i)
		=\\&=
		\sum_{q\in Q}(\bz(q)\cdot \next_p(q)-\bz(q)\cdot \from_p(q))  + \sum_{i=1}^{d} \by(i)\cdot \effect_p (i)
		\leq 0
	\end{align*}
	Therefore \(\by,\bz\) is also a solution to (II')
	
\end{proof}

We will now rewrite the constraint systems (I') and (II') into matrix form. Let \(D\) be a \(A\times \{1,\dots, d \} \) matrix whose columns are indexed by elements of \(A \), and rows indexed by counters \(c\in \{1,\dots,d\} \), such that the column \(D(a)=\effect_a \). And let \(F \) be a \(A\times Q \) matrix, whose columns are indexed by elements of \(A \), and rows are indexed by states \(p\in Q \), such that the column \(F(a) \) is equal to the vector \(\bw \) such that \(\bw(p)=\next_a(p)-\from_a(p) \) for each \(p\in Q\).

Then we can further rewrite the systems (I') and (II') as follows:

%
%

\vspace{0.7cm}
\begin{tabular}{|c|c|}
	\hline
	{\begin{minipage}[c]{0.45\linewidth}
			\vspace{0.2cm}
			constraint system (I'):
			
			\vspace{0.2cm}
			Find $\bx' \in \mathbb{Z}^{A}$ such that
			\begin{align}
				D \bx' & \ge \vec{0} \nonumber\\
				\bx' & \ge \vec{0} \nonumber\\
				F \bx' & = \vec{0} \nonumber
			\end{align}
			
			\vspace{0.2cm}
	\end{minipage}}
	&
	{\begin{minipage}[c]{0.45\linewidth}
			\vspace{0.2cm}
			constraint system (II'):
			
			\vspace{0.2cm}
			Find $\by \in \mathbb{Z}^d,\bz \in \mathbb{Z}^{Q}$ with
			\begin{align}
				\by & \ge  \vec{0} \nonumber\\
				\bz & \ge  \vec{0} \nonumber\\
				F^T \bz+ D^T \by  & \le \vec{0} \nonumber
			\end{align}

			\vspace{0.2cm}
	\end{minipage}}\\
	\hline
\end{tabular}
\vspace{0.7cm}

The rest then follows exactly the same as the the proof of the dichotomy on non-stochastic VASS in \cite{Zuleger:VASS-polynomial} (Lemma~4), as  the only difference between our systems and the ones used in \cite{Zuleger:VASS-polynomial} is that the matrix \(F \) now also may contain rational numbers other than \(-1,0,1 \). The proof in \cite{Zuleger:VASS-polynomial} is already made over \(\mathbb{Z} \), and the only additional requirement it needs is that each column of \(F \) sums up to \(0 \), which is satisfied also by our \(F \). 

\subsection{The proof from \cite{Zuleger:VASS-polynomial} (Lemma~4)}
\begin{disclaimer*}
For the sake of completeness we include a copy of the proof from \cite{Zuleger:VASS-polynomial} (Lemma~4). All credit for the proof in this subsection goes to the author of \cite{Zuleger:VASS-polynomial}. The only changes we made was to rename some variables.
	\end{disclaimer*}

The proof will be obtained by two applications of Farkas' Lemma.
We will employ the following version of Farkas' Lemma, which states that for matrices $A$,$C$ and vectors $b$,$d$, exactly one of the following statements is true:

\vspace{0.3cm}
\begin{tabular}{|c|c|}
	\hline
	\begin{minipage}[c]{0.4\linewidth}
			
			\vspace{0.2cm}
			there exists $x$ with
			
			\vspace{0.2cm}
			$\begin{array}{rcr}
					Ax & \ge & b \\
					Cx & = & d
				\end{array}$
			
		\end{minipage}
	&
	\begin{minipage}[c]{0.5\linewidth}
			\vspace{0.2cm}
			
			\vspace{0.2cm}
			there exist $y,z$ with
			
			\vspace{0.2cm}
			$\begin{array}{rcr}
					y & \ge & 0 \\
					A^T y + C^T z & = & 0 \\
					b^T y + d^T z & > & 0
				\end{array}$
			\vspace{0.2cm}
		\end{minipage} \\
	
	\hline
\end{tabular}
\vspace{0.2cm}

We now consider the constraint systems~($\conSysA_\transition$) and~($\conSysB_\transition$) stated below.
Both constraint systems are parameterized by $a \in A$ (we note that only Equations (\ref{multcycle:eq4}) and (\ref{multcycle:eq5}) are parameterized by $a$).

\vspace{0.3cm}
\begin{tabular}{|c|c|}
	\hline
	{\begin{minipage}[c]{0.4\linewidth}
					\vspace{0.2cm}
					constraint system ($\conSysA_\transition$):
					
					\vspace{0.2cm}
					there exists $\counters \in \mathbb{Z}^{A}$ with
					\begin{align}
							U \counters & \ge 0 \nonumber\\
							\counters & \ge 0 \nonumber\\
							\flowMatrix \counters & = 0 \nonumber\\
							\counters(\transition) & \ge 1 \label{multcycle:eq4}
						\end{align}
					\vspace{-0.5cm}
			\end{minipage}}
	&
	{\begin{minipage}[c]{0.52\linewidth}
					constraint system ($\conSysB_\transition$):
					
					\vspace{0.2cm}
					there exist\\
					$\rankCoeff \in \mathbb{Z}^\vars,\offsets \in \mathbb{Z}^{\states(\vass)}$ with
					\begin{align}
							\rankCoeff & \ge  0 \nonumber\\
							\offsets & \ge  0 \nonumber\\
							\updates^T \rankCoeff + \flowMatrix^T \offsets & \le 0 \text{ with } < 0 \text{ in line } \transition \label{multcycle:eq5}
						\end{align}
			\end{minipage}}\\
	\hline
\end{tabular}
\vspace{0.3cm}

We recognize constraint system~($\conSysA_\transition$) as the dual of constraint system~($\conSysB_\transition$)
in the following Lemma:

\begin{lemma}
	\label{lem:ranking-or-witness}
	Exactly one of the constraint systems ($\conSysA_\transition$) and ($\conSysB_\transition$) has a solution.
\end{lemma}
\begin{proof}
	We fix some $a\in A$.
	We denote by $\character_a \in \mathbb{Z}^{A}$ the vector with $\character_a(a') = 1$, if $a' = a$, and $\character_a(a') = 0$, otherwise.
	Using this notation we rewrite ($\conSysA_\transition$) to the equivalent constraint system ($\conSysA_\transition'$):
	
	\vspace{0.2cm}
	\begin{tabular}{|lc|}
			\hline
			{\begin{minipage}[r]{0.3\linewidth}
							constraint system ($\conSysA_\transition'$):
							\vspace{0.8cm}
					\end{minipage}}
			&
			{\begin{minipage}[r]{0.3\linewidth}
							\vspace{-0.2cm}
							\begin{eqnarray*}
									\begin{pmatrix}
											\updates \\
											\identity
										\end{pmatrix}
									\counters & \ge &
									\begin{pmatrix}
											0 \\
											\character_\transition
										\end{pmatrix}\\
									\flowMatrix \counters & = & 0
								\end{eqnarray*}
							\vspace{-0.6cm}
					\end{minipage}}\\
			\hline
		\end{tabular}
	\vspace{0.2cm}
	
	Using Farkas' Lemma, we see that either ($\conSysA_\transition'$) is satisfiable or the following constraint system ($\conSysB_\transition'$) is satisfiable:
	
	\vspace{0.3cm}
	\begin{tabular}{|r|c|}
			\hline
			{\begin{minipage}[r]{0.47\linewidth}
							\vspace{0.2cm}
							\hspace{-0.15cm}
							constraint system ($\conSysB_\transition'$):
							
							\vspace{-0.7cm}
							\begin{eqnarray*}
									\begin{pmatrix}
											\rankCoeff \\
											k
										\end{pmatrix} & \ge & 0 \\
									\begin{pmatrix}
											\updates \\
											\identity
										\end{pmatrix}^T
									\begin{pmatrix}
											\rankCoeff \\
											k
										\end{pmatrix} + \flowMatrix^T \offsets & = & 0 \\
									\begin{pmatrix}
											0 \\
											\character_\transition
										\end{pmatrix}^T
									\begin{pmatrix}
											\rankCoeff \\
											k
										\end{pmatrix} + 0^T \offsets & > & 0
								\end{eqnarray*}
							\vspace{-0.3cm}
					\end{minipage}}
			&
			{\begin{minipage}[c]{0.45\linewidth}
							constraint system ($\conSysB_\transition'$) simplified:
							\begin{eqnarray*}
									\rankCoeff & \ge & 0 \\
									k & \ge & 0 \\
									\updates^T \rankCoeff + k + \flowMatrix^T \offsets & = & 0 \\
									k(\transition) & > & 0
								\end{eqnarray*}
					\end{minipage}}\\
			\hline
		\end{tabular}
	\vspace{0.3cm}
	
	We observe that solutions of constraint system ($\conSysB_\transition'$) are invariant under shifts of $\offsets$, i.e, if $\rankCoeff$, $k$, $\offsets$ is a solution, then $\rankCoeff$, $k$, $\offsets + c \cdot \oneVec$ is also a solution for all $c \in \mathbb{Z}$  (because elements of every row of $\flowMatrix^T$ sum up to \(0\)).
	Hence, we can force $\offsets$ to be non-negative.
	We recognize that constraint systems ($\conSysB_\transition'$) and ($\conSysB_\transition$) are equivalent. \qed
\end{proof}

We now consider the constraint systems ($\conSysC_\var$) and ($\conSysD_\var$) stated below.
Both constraint systems are parameterized by a counter $\var $ (we note that only Equations (\ref{multcycle:eq6}) and (\ref{multcycle:eq7}) are parameterized by $\var$).

\vspace{0.3cm}
\begin{tabular}{|c|c|}
	\hline
	{\begin{minipage}[c]{0.42\linewidth}
					\vspace{0.2cm}
					constraint system ($\conSysC_\var$):
					
					\vspace{0.2cm}
					there exists $\counters \in \mathbb{Z}^{A}$ with
					\begin{align}
							\updates \counters & \ge 0 \text{ with } \ge 1 \text{ in line } \var \label{multcycle:eq6}\\
							\counters & \ge 0 \nonumber\\
							\flowMatrix \counters & = 0 \nonumber
						\end{align}
					\vspace{-0.5cm}
			\end{minipage}}
	&
	{\begin{minipage}[c]{0.5\linewidth}
					\vspace{0.2cm}
					constraint system ($\conSysD_\var$):
					
					\vspace{0.2cm}
					there exist
					$\rankCoeff \in \mathbb{Z}^\vars,\offsets \in \mathbb{Z}^{\states(\vass)}$ with
					\begin{align}
							\rankCoeff & \ge  0 \nonumber\\
							\offsets & \ge  0 \nonumber\\
							\updates^T \rankCoeff + \flowMatrix^T \offsets & \le 0 \nonumber\\
							\rankCoeff(\var) & > 0 \label{multcycle:eq7}
						\end{align}
					\vspace{-0.5cm}
			\end{minipage}}\\
	\hline
\end{tabular}
\vspace{0.3cm}

We recognize constraint system ($\conSysC_\var$) as the dual of constraint system ($\conSysD_\var$) in the following Lemma:

\begin{lemma}
	\label{lem:ranking-or-witness-two}
	Exactly one of the constraint systems ($\conSysC_\var$) and ($\conSysD_\var$) has a solution.
\end{lemma}
\begin{proof}
	We fix some counter~$\var$.
	We denote by $\character_\var \in \mathbb{Z}^{\vars}$ the vector with $\character_\var(\var') = 1$, if $\var' = \var$, and $\character_\var(\var') = 0$, otherwise.
	Using this notation we rewrite ($\conSysA_\var$) to the equivalent constraint system ($\conSysA_\var'$):
	
	\vspace{0.2cm}
	\begin{tabular}{|lc|}
			\hline
			{\begin{minipage}[r]{0.3\linewidth}
							constraint system ($\conSysC_\var'$):
							\vspace{0.8cm}
					\end{minipage}}
			&
			{\begin{minipage}[r]{0.3\linewidth}
							\vspace{-0.2cm}
							\begin{eqnarray*}
									\begin{pmatrix}
											\updates \\
											\identity
										\end{pmatrix}
									\counters & \ge &
									\begin{pmatrix}
											\character_\var \\
											0
										\end{pmatrix}\\
									\flowMatrix \counters & = & 0
								\end{eqnarray*}
							\vspace{-0.6cm}
					\end{minipage}}\\
			\hline
		\end{tabular}
	\vspace{0.2cm}
	
	Using Farkas' Lemma, we see that either ($\conSysC_\var'$) is satisfiable or the following constraint system ($\conSysD_\var'$) is satisfiable:
	
	\vspace{0.3cm}
	\begin{tabular}{|r|c|}
			\hline
			{\begin{minipage}[r]{0.47\linewidth}
							\vspace{0.2cm}
							\hspace{-0.15cm}
							constraint system ($\conSysD_\var'$):
							
							\vspace{-0.7cm}
							\begin{eqnarray*}
									\begin{pmatrix}
											\rankCoeff \\
											k
										\end{pmatrix} & \ge & 0 \\
									\begin{pmatrix}
											\updates \\
											\identity
										\end{pmatrix}^T
									\begin{pmatrix}
											\rankCoeff \\
											k
										\end{pmatrix} + \flowMatrix^T \offsets & = & 0 \\
									\begin{pmatrix}
											\character_\var \\
											0
										\end{pmatrix}^T
									\begin{pmatrix}
											\rankCoeff \\
											k
										\end{pmatrix} + 0^T \offsets & > & 0
								\end{eqnarray*}
							\vspace{-0.3cm}
					\end{minipage}}
			&
			{\begin{minipage}[c]{0.45\linewidth}
							constraint system ($\conSysB_\transition'$) simplified:
							\begin{eqnarray*}
									\rankCoeff & \ge & 0 \\
									k & \ge & 0 \\
									\updates^T \rankCoeff + k + \flowMatrix^T \offsets & = & 0 \\
									\rankCoeff(\var) & > & 0
								\end{eqnarray*}
					\end{minipage}}\\
			\hline
		\end{tabular}
	\vspace{0.3cm}
	
	We observe that solutions of constraint system ($\conSysD_\var'$) are invariant under shifts of $\offsets$, i.e, if $\rankCoeff$, $k$, $\offsets$ is a solution, then $\rankCoeff$, $k$, $\offsets + c \cdot \oneVec$ is also a solution for all $c \in \mathbb{Z}$  (because elements of every row of $\flowMatrix^T$ sum up to \(0\)).
	Hence, we can force $\offsets$ to be non-negative.
	We recognize that constraint systems ($\conSysD_\var'$) and ($\conSysD_\var$) are equivalent. \qed
\end{proof}

\subsection{Proof of Lemma~\ref{lemma-linear-upper-bound}}
\label{app-lem:linear-upper-bound-proof}

\begin{lemma*}[\textbf{\ref{lemma-linear-upper-bound}}]
	For every counter $c$ such that $\by(c)>0$, every $\varepsilon >0$, every $p\in Q$, and every $\sigma \in \Sigma$,  there exists $n_0$ such that for all $n\geq n_0$ we have that $\prob^{\sigma}_{p \bn}(\calC[c] \geq n^{1+\varepsilon} )\leq kn^{-\varepsilon}$ where~$k$ is a constant depending only on~$\A$. 
\end{lemma*}
	Let \(P_1V_1,P_2V_2,\dots \) be the random variables encoding the computation under \(\sigma \) from \(p\bn\) (i.e. \(P_iV_i \) represents the configuration at \(i\)-th step of the computation). And let \(R_1,R_2,\dots \) represent the value of \(rank \) at \(i-\)th step (i.e. \(R_i=rank(P_iV_i) \)). Then \(R_1,R_2,\dots \) is a supermartingale.
	
	\begin{lemma}
		\(R_1,R_2,\dots \)  is a supermartingale.
	\end{lemma}
	\begin{proof}
		One can express \(R_{i+1}=R_{i}+X_{i+1} \), where \(X_{i+1}=R_{i+1}-R_{i} \) is the change of \(rank \) in the \((i+1)\)-st step. Then it holds \(\Exp^\sigma_{p\vec{n}}(R_{i+1}|R_i)=\Exp^\sigma_{p\vec{n}}(R_{i}|R_i)+\Exp^\sigma_{p\vec{n}}(X_{i+1}|R_i)=R_i+\Exp^\sigma_{p\vec{n}}(X_{i+1}|R_i) \). We want to show that \(\Exp^\sigma_{p\vec{n}}(X_{i+1}|R_i)\leq 0 \). Let \(T_{i+1} \) be random variable representing the transition taken at \((i+1)\)-st step. Then \(\Exp^\sigma_{p\vec{n}}(X_{i+1}|R_i)= \sum_{t\in T} \prob_{p\bn}^\sigma(T_{i+1}=t|R_i)\cdot \RankEff(t) \) where \(\RankEff(t) \) represents the change of \(rank\) under transition \(t \).

		Let \(T_n=\bigcup_{p\in Q_n} \tout(p)\) and \(T_p=\bigcup_{p\in Q_p} \tout(p) \) , then we can write		
		\(\Exp^\sigma_{p\vec{n}}(X_{i+1}|R_i)=
		\sum_{t\in T_p} \prob_{p\bn}^\sigma(T_{i+1}=t|R_i)\cdot \RankEff(t)+\sum_{t\in T_n} \prob_{p\bn}^\sigma(T_{i+1}=t|R_i)\cdot \RankEff(t)\).

		Since for each \(t=(p,\bu,q)\in T \) it holds \(\RankEff(t)=\bz(q)-\bz(p)+\sum_{i=1}^d \by(i)\bu(i)\), for each \(t\in T_n \) it holds \(\RankEff(t)\leq 0 \),  and for each \(p\in Q_n \), it holds \(\sum_{t\in \tout(p)} P(t)\RankEff(t) \leq 0 \). Therefore we can write 
		\[ \sum_{t\in T_p} \prob_{p\bn}^\sigma(T_{i+1}=t|R_i)\cdot \RankEff(t)
		=
		\sum_{p\in Q_p} (\prob_{p\bn}^\sigma(P_{i}=p) \sum_{t\in \tout(p)} P(t)\cdot \RankEff(t))
		\leq 0		
		 \]
		and 
		\[\sum_{t\in T_n} \prob_{p\bn}^\sigma(T_{i+1}=t|R_i)\cdot \RankEff(t)\leq 0
		 \]
		 Thus \(E(X_{i+1}|R_i)\leq 0\)

	\end{proof}

	Now let us consider the stopping rule \(\tau \) that stops when either any counter reaches \(0 \), or any counter \(c\) with \(\by(c)>0 \) becomes larger then \(n^{1+\epsilon} \) for the first time.	(i.e. either \(V_\tau(c')<0 \) for any \(c'\in \{1,\dots,d \} \), or \(V_\tau(c)\geq n^{1+\epsilon} \) for \(c \) with \(\by(c)>0 \)). Then for all \(i\), it holds that \(R_{min(i,\tau)}\leq max_{p\in Q} \bz(p) + max_{c\in\{1,\dots,d\}} \by(c)\cdot d \cdot (n^{1+\epsilon}+u)  \), where \(u \) is the maximal increase of a counter in a single transition. Therefore we can apply optional stopping theorem to obtain:
	\[ max_{p\in Q} \bz(p) + max_{c\in\{1,\dots,d\}} \by(c)\cdot d\cdot  n \geq E(R_1)\geq E(R_\tau)\geq   pX_{n^{1+\epsilon}}+(1-p)X_0 \]
	where \(X_{n^{1+\epsilon}} \) represents the minimal possible value of \(R_\tau \) if any counter \(c\) with \(\by(c)>0 \) has  \(R_\tau(c)\geq n^{1+\epsilon} \), \(p \) is the probability of any such counter being at least \(n^{1+\epsilon} \) upon stopping, and \(X_0 \) represents the minimal value of \(R_\tau \) if no such counter reached \(n^{1+\epsilon} \). We can simplify this as
	\[max_{p\in Q} \bz(p) + max_{c\in\{1,\dots,d\}} \by(c)\cdot d\cdot  n \geq   pX_{n^{1+\epsilon}}+(1-p)X_0 \]
	\[max_{p\in Q} \bz(p) + max_{c\in\{1,\dots,d\}} \by(c)\cdot d\cdot  n-(1-p)\cdot X_0 \geq   pX_{n^{1+\epsilon}} \]
	\[max_{p\in Q} \bz(p) + max_{c\in\{1,\dots,d\}} \by(c)\cdot d\cdot  n-(1-p)max_{c\in\{1,\dots,d\}} \by(c) \cdot d\cdot u \geq   pX_{n^{1+\epsilon}} \]
	\[max_{p\in Q} \bz(p) + max_{c\in\{1,\dots,d\}} \by(c)\cdot d\cdot  n-max_{c\in\{1,\dots,d\}} \by(c) \cdot d\cdot u +p\cdot  max_{c\in\{1,\dots,d\}} \by(c) \cdot d\cdot u \geq   pX_{n^{1+\epsilon}} \]
	\[max_{p\in Q} \bz(p) + max_{c\in\{1,\dots,d\}} \by(c)\cdot d\cdot  n-max_{c\in\{1,\dots,d\}} \by(c) \cdot d\cdot u \geq   pX_{n^{1+\epsilon}}-p \cdot max_{c\in\{1,\dots,d\}} \by(c) \cdot d\cdot u \]	
	\[max_{p\in Q} \bz(p) + max_{c\in\{1,\dots,d\}} \by(c)\cdot d\cdot  n-max_{c\in\{1,\dots,d\}} \by(c) \cdot d\cdot u \geq   p(X_{n^{1+\epsilon}}- max_{c\in\{1,\dots,d\}} \by(c) \cdot d\cdot u) \]	\[\frac{max_{p\in Q} \bz(p) + max_{c\in\{1,\dots,d\}} \by(c)\cdot d\cdot  n-max_{c\in\{1,\dots,d\}} \by(c) \cdot d\cdot u }{X_{n^{1+\epsilon}}- max_{c\in\{1,\dots,d\}} \by(c) \cdot d\cdot u} \geq   p \]	
	\[\frac{max_{p\in Q} \bz(p) + max_{c\in\{1,\dots,d\}} \by(c)\cdot d\cdot  n-max_{c\in\{1,\dots,d\}} \by(c) \cdot d\cdot u }{n^{1+\epsilon}- max_{c\in\{1,\dots,d\}} \by(c) \cdot d\cdot u} \geq   p \]
	
		As for all sufficiently large \(n\) it holds \(0.5\cdot n^{1+\epsilon}\leq n^{1+\epsilon}- max_{c\in\{1,\dots,d\}} \by(c) \cdot d\cdot u \)  we have
	
	\[\frac{max_{p\in Q} \bz(p) + max_{c\in\{1,\dots,d\}} \by(c)\cdot d\cdot  n-max_{c\in\{1,\dots,d\}} \by(c) \cdot d\cdot u }{0.5\cdot n^{1+\epsilon}} \geq   p \]
	\[\frac{max_{p\in Q} \bz(p)-max_{c\in\{1,\dots,d\}} \by(c) \cdot d\cdot u}{0.5\cdot n^{1+\epsilon}} + \frac{max_{c\in\{1,\dots,d\}} \by(c)\cdot d\cdot n}{0.5\cdot n^{1+\epsilon}} \geq p \]
	\[\frac{max_{p\in Q} \bz(p)-max_{c\in\{1,\dots,d\}} \by(c) \cdot d\cdot u}{0.5\cdot n^{1+\epsilon}} + \frac{max_{c\in\{1,\dots,d\}} \by(c)\cdot d}{0.5\cdot n^{\epsilon}} \geq p \]
	
	Also as \(n^{1+\epsilon}\geq n^{\epsilon} \) we  have

	\[\frac{max_{p\in Q} \bz(p)-max_{c\in\{1,\dots,d\}} \by(c) \cdot d\cdot u}{0.5\cdot n^{\epsilon}} + \frac{max_{c\in\{1,\dots,d\}} \by(c)\cdot d}{0.5\cdot n^{\epsilon}} \geq p \]
\[\frac{max_{p\in Q} 2\cdot\bz(p)-max_{c\in\{1,\dots,d\}} 2\cdot\by(c) \cdot d\cdot u + max_{c\in\{1,\dots,d\}} 2\cdot\by(c)\cdot d}{ n^{\epsilon}} \geq p \]
	
	

	
	
	
	As \(k=max_{p\in Q} 2\cdot\bz(p)-max_{c\in\{1,\dots,d\}} 2\cdot\by(c) \cdot d\cdot u + max_{c\in\{1,\dots,d\}} 2\cdot\by(c)\cdot d  \) is a constant dependent only on the VASS MDP, it holds for each counter \(c\) with \( \by(c)>0\) and for all sufficiently large \(n\) that \(\prob^{\sigma}_{p \bn}(\calC[c] \geq n^{1+\epsilon} )\leq p \leq kn^{-\epsilon}\).

\subsection{Proof of Lemma~\ref{lemma-quadratic-lower-bound}}
\label{app-lem:-quadratic-lower-bound-proof}

\begin{lemma*}[\textbf{\ref{lemma-quadratic-lower-bound}}]
	For each counter $c$ such that $\by(c)=0$ we have that $\calC_{\eexp}[c] \in \Omega(n^2)$ and $n^2$ is a lower estimate of $\calC[c]$. Furthermore, for every $\varepsilon>0$ there exist a sequence of strategies $\sigma_1,\sigma_2,\dots $, a constant $k$, and  $p\in Q$ such that for every $0<\varepsilon'<\varepsilon$, we have that 
	\[
	\lim_{n \to \infty} \prob_{p\bn}^{\sigma_n}(\Reach^{\leq kn^{2-\varepsilon'}}(\mathit{Target}_n)) = 1
	\]
	where $\mathit{Target}_n = \{q\bv \in \conf(\A) \mid \bv(c)\geq n^{2-\varepsilon} \mbox{ for every counter $c$ such that $\by(c)=0$}\}$.
	%
\end{lemma*}

Let \(\A_\bx \) be the VASS MDP induced by transitions \(t \) with \(\bx(t)>0 \).

\begin{lemma}
	In \(\A_\bx \), Each pair of states \( p,q\in Q \) is either a part of the same MEC of \(\A_\bx\), or \(p\) is not reachable from \(q\) and vice-versa, in \(\A_\bx \).
\end{lemma}

\begin{proof}
	This follows directly from \(\bx \) satisfying kirhoff laws. 	
\end{proof}

Therefore \(\A_\bx \) can be decomposed into multiple MECs, and there are no transitions in the MEC decomposition of \(\A_\bx \). Let these MECs be \(B_1,\dots,B_k \), and let \(\bx_1,\dots,\bx_k \) be the restriction of \(\bx \) to the transitions of \(B_1,\dots, B_k \). (i.e. \(\bx_i(t)=\bx(t)\) if \(B_i \) contains \(t \), and otherwise \(\bx_i(t)=0\)). 

For each \(1\leq i\leq k \), let  \(\overline{\bx}_i=\frac{\bx_i}{\sum_{t\in T} \bx_i(t)} \) be the normalized vector of \(\bx_i \), and let \(\sigma_i \) be a Markovian strategy for \(B_i \) such that \(\sigma_i(p)(t)=\frac{\overline{\bx}_i(t)}{\sum_{t\in \tout(p)} \overline{\bx}_i(t)} \) for \(t\) such that \(\sum_{t\in \tout(p)} \overline{\bx}_i(t)>0 \), and undefined otherwise. We will use \(M_i \) to represent the Markov chain obtained by applying \(\sigma_i \) to \( B_i\).



\begin{lemma}\label{app-lemma-solution-I-startegy}
	Let \(m_i\in \mathbb{Z}^Q \) be such that \(m_i(p)=\sum_{t\in \tout(p)} \overline{\bx}_i(t) \). Then \(m_i \) is an invariant distribution on \(M_i \). Also the expected effect of a single computational step in \(M_i \) taken from distribution \(m_i \) is equal to \(\sum_{(p,\bu,q)\in T} \overline{\bx}_i((p,\bu,q))\bu \). 
\end{lemma}
\begin{proof}
	Let us consider a single computation step in \(M_i \) taken from distribution \( m_i\), and let \(X \) be resulting distribution on transitions during this step. Then for each transition \(t=(p,\bu,q)\in T \) it holds:	 
	\begin{itemize}
		\item if \(p\in Q_p \) then \( X(t)=P(t)\cdot m_i(p) =P(t)\cdot \sum_{t\in \tout(p)} \overline{\bx}_i(t)=\overline{\bx}_i(t) \).
		\item if \(p\in Q_n \) then \( X(t)=
		\sigma_i(p)(t)\cdot m_i(p)
		=
		\frac{\overline{\bx}_i(t)}{\sum_{t\in \tout(p)} \overline{\bx}_i(t)}\cdot  \sum_{t\in \tout(p)} \overline{\bx}_i(t) 
		=
		\overline{\bx}_i(t)
		\).
		
	\end{itemize} 
	And as the next distribution \(m_i' \) on states can be expressed as \(m_i'(p)=\sum_{(q,\bu,p)\in T} X((q,\bu,p)) = \sum_{(q,\bu,p)\in T} \overline{\bx}_i((q,\bu,p))=\sum_{t\in \tout(p)}\overline{\bx}_i(t)=m_i(p) \), \(m_i \) is an invariant distribution on \( M_i\).
	
\end{proof}

Let \(\bj_1,\dots, \bj_k \) be the expected update vectors per single computational step generated by the invariants in \(M_1,\dots, M_k \). Then from \(x \) being a solution to (I), we get that for \( a_i= \sum_{t\in T} x_i(t)  \) it holds \[\sum_{i=1}^{k} a_i\cdot \bi_i \geq \vec{0} \] as well as \[(\sum_{i=1}^{k} a_i\cdot \bi_i)(c) > 0 \] for \(c \) with \(\by(c)=0 \).

Therefore we can use the results of \cite{BCKNV:probVASS-linear-termination}, which states that if there exists a sequence of Markov chains \(M_1,\dots,M_k \) with their respective increments \(\bj_1,\dots,\bj_k \),  and positive integer coefficients \(a_1,\dots,a_k \) such that \(\sum_{i=1}^k a_i\bj_i\geq \vec{0} \), then there exists a function \(L(n)\in \Theta(n) \), a state \(p\in Q\), and sequence of strategies \(\sigma_1,\sigma_2,\dots \) such that the probability \(X_n \) of the computation from \(p\bn \) under \(\sigma_n \) never decreasing at each of the first \(L^{2-\epsilon'}(n) \) steps any counter below  \(b_1n\Exp_{p\vec{n}}^{\sigma_n}(C_i^n(c))-b_2n \), where \(b_1,b_2\) are some constants and \(C_i^n\) is the random variable representing the counter vector after \(i\) steps when computing form \(p\bn\) under \(\sigma_n\), satisfies \(\lim_{n\rightarrow \infty} X_n = 1\). And furthermore,  for each counter \(c\) with \((\sum_{i=1}^k a_i\bj_i)(c) > 0 \) it holds that \(\Exp_{p\vec{n}}^{\sigma_n}(C_i^n(c))\in \Omega(i)\). 

Therefore, with probability at least \(X_n\) we reach a configuration \(q\bv\) with each counter \(c\) such that \(\by(c)=0 \) having \(\bv(c)\geq n^{2-\epsilon} \) within  \(L^{2-\epsilon'}(n)\leq kn^{2-\epsilon'}\) steps, and it holds \(\lim_{n\rightarrow\infty } X_n=1 \).

\subsection{VASS MDP with DAG-like MEC Decomposition}
\label{app-DAG-like}


We formalize and prove the idea sketched at the end of Section~\ref{section-linquad}. 

\begin{lemma}\label{lemma-linquad-for-DAG}
Let $\A$ be a DAG-like VASS MDP with $d$ counters and a DAG-like MEC decomposition, and \( \beta=M_1,\dots, M_k \) be it's type. Let \(\bw_0,\bw_1,\dots,\bw_k\in \{n,\infty \}^d \), and let \(M_i^{\bw_{i-1}} \) be the MEC obtained by taking \(M_i\) and changing the effect \(\bu \) of every transition to \(\bu' \) such that for each \(c\in\{1,\dots,d \} \), \(\bu'(c)=\bu(c) \) if \(\bw_{i-1}(c)=n \), and \(\bu'(c)=0 \) if \(\bw_{i-1}(c)=\infty \). Furthermore, let the following hold for each counter \(c\in\{1,\dots,d \}\) and \(1\leq i\leq k \)
\begin{itemize}
	\item \(\bw_0(c)=n \),
	\item \(\bw_i(c)=n \) if both \(\bw_{i-1}(c)=n \) and \(n \) is a tight estimate of \(c\) in \(M_i^{\bw_{i-1}} \),
	\item \(\bw_i(c)=\infty \) if either \(\bw_{i-1}(c)=\infty \) or \(n^2 \) is a lower estimate of \(c\) in \(M_i^{\bw_{i-1}} \).
\end{itemize}
Then for each \(\epsilon>0 \), there exists a sequence of strategies \(\sigma_1^1,\sigma_1^2,\dots,\sigma_1^k,\sigma_2^1,\dots,\sigma_2^k,\sigma_3^1,\dots \), such that for each \(1\leq i\leq k \), and each \(n\), the computation under \(\sigma_n^i \) initiated in some state of \(M_i \) with initial counter vector \(\bv \) such that for each \(c\in\{1,\dots,d \}\) it holds
 \begin{itemize}
	\item \(\bv(c)\geq \lfloor\frac{n}{2^{i}}\rfloor \) if \(\bw_{i-1}(c)=n \),
	\item \(\bv(c)\geq \lfloor(\frac{n}{2^{i}})^{2-\epsilon_{i-1}}/2^{i}\rfloor \) if \(\bw_{i-1}(c)=\infty \),
\end{itemize}
   reaches with probability \(X_n\) a configuration of \(M_i\) with counter vector \(\bu \) such that for each \(c\in\{1,\dots,d \}\) it holds
  \begin{itemize}
  	\item \(\bu(c)\geq \lfloor\frac{n}{2^{i+1}}\rfloor \) if \(\bw_{i}(c)=n \),
  \item \(\bu(c)\geq \lfloor(\frac{n}{2^{i+1}})^{2-\epsilon_{i}}/2^{i+1}\rfloor \) if \(\bw_{i}(c)=\infty \),
  \end{itemize} 
where \(\epsilon_i=\frac{i\epsilon}{k}\), and it holds \(\lim_{n\rightarrow\infty} X_n=1 \). 
Furthermore, for each counter \(c\in\{1,\dots,d \}\), if \(\bw_k(c)=n \) then \(n \) is a tight estimate of \(\calC[c] \) for type \(\beta \), and if \(\bw_k(c)=\infty \) then \(n^2 \) is a lower estimate of \(\calC[c] \) for type \(\beta \). 
	\end{lemma}

\begin{proof}
	Proof by induction on \(k\). 
	Base case of \(k=1\) holds from Lemma~\ref{lemma-quadratic-lower-bound}, and the second part holds from Lemma~\ref{lemma-C-limPlinear-if-no-+}.
	 Assume now the Lemma holds for the type \(M_1,\dots, M_{i-1} \). Let \(\sigma_1,\sigma_2,\dots \) and \(p\in Q\) be from the Lemma~\ref{lemma-quadratic-lower-bound} for \(M_i^{\bw_{i-1}} \) and for \(\epsilon_i \). Then from induction assumption, there are strategies such that when the computation reaches \(M_i \) the counters vector is \(\bv \) with probability \(Y_i\) such that \(\lim_{n\rightarrow\infty} Y_n=1 \) and 
	  \begin{itemize}
	 	\item \(\bv(c)\geq \lfloor\frac{n}{2^{i}}\rfloor \) if \(\bw_{i-1}(c)=n \),
	 	\item \(\bv(c)\geq \lfloor(\frac{n}{2^{i}})^{2-\epsilon_{i-1}}/2^{i}\rfloor \) if \(\bw_{i-1}(c)=\infty \).
	 \end{itemize} 
	Now let us consider the following: upon reaching \(M_i\), we divide the counters vector \(\bv \) into two halves, each of size \(\lfloor\frac{\bv}{2}\rfloor \), and then we perform the computation of \(\sigma_{\lfloor\frac{n}{2^{i+1}}\rfloor} \) on the first half for \(ln^{2-\frac{(i+0.5)\epsilon}{k}}\) steps. (i.e., if the effect on any counter \(c\) is less then \(-\lfloor\frac{\bv}{2}\rfloor(c) \), then the computation stops). Then from Lemma~\ref{lemma-quadratic-lower-bound} we will with probability \(X_n\) reach a configuration \(\bu\) with all counters  \(c\) such that \(\bw_{i-1}(c)=n \) and \(\bw_{i}(c)=\infty \) being at least \(\bu(c)\geq (\lfloor\frac{n}{2^{i+1}}\rfloor)^{2-\epsilon_i} \), such that \(\lim_{n\rightarrow \infty} X_n=1 \). As the length of this computation is only \(ln^{2-\frac{(i+0.5)\epsilon}{k}}\) we cannot decrease any "deleted" counter \(c\) with \(\bw_{i-1}(c)=\infty \) by more then \(an^{2-\frac{(i+0.5)\epsilon}{k}} \) for some constant \(a\). Therefore for all sufficiently large \(n\), the computation cannot terminate due to such counter being depleted.  And since the second half of \(\bv \) is untouched, we still have for each counter \(c\) with \(\bw_{i-1}(c)=\infty\) at least \(\lfloor(\frac{n}{2^{i}})^{2-\epsilon_{i-1}}/2^{i+1}\rfloor \) and for each counter \(c \) with \(\bw_{i-1}(c)=n\) at least \(\lfloor\frac{n}{2^{i+1}}\rfloor \). 

Therefore with probability at least \(X_nY_n \) the computation ends in configuration \(q\bu\) of \(M_i \) such that for each counter \(c\)
  \begin{itemize}
	\item \(\bu(c)\geq \lfloor\frac{n}{2^{i+1}}\rfloor \) if \(\bw_{i}(c)=n \),
	\item \(\bu(c)\geq \lfloor(\frac{n}{2^{i+1}})^{2-\epsilon_{i}}/2^{i+1}\rfloor \) if \(\bw_{i}(c)=\infty \)
\end{itemize} 
And it holds \(\lim_{n\rightarrow \infty} X_nY_n=1 \). Thus \(n^2 \) is a lower estimate of \(\calC[c] \) for type \(M_1,\dots,M_i \) for each \(c\) with \(\bw_i(c)=\infty \).
	
	For the second part of the lemma, let \(\sigma \) be some strategy. Then for every counter \(c\) with \(\bw_i(c)=n \), we have from the induction assumption that the probability, of the strategy \(\sigma \) started in initial configuration \(p\vec{n}\), reaching \(M_i \) along type \(M_1,\dots,M_i \) with \(c\) being at least \(n^{1+\epsilon'} \) for some \(0< \epsilon' \), is at most \(Z_n \), where \(\limsup_{n \to \infty} Z_n=0 \). Let \(\sigma' \)be a strategy, which for an initial state \(q\) such that \(q\) is a state of \(M_1 \), computes as \(\sigma \) after a path from \(p \) to \(q\). Then from Lemma~\ref{lemma-linear-upper-bound} we have for each \(0<\hat\epsilon \)  that  \begin{multline*}
	\prob^{\sigma'}_{q n^{1+\epsilon'}}(\calC_{M_i^{\bw_{i-1}}}[c] \geq n^{1+\hat\epsilon} )
	=
	\prob^{\sigma'}_{q n^{1+\epsilon'}}(\calC_{M_i^{\bw_{i-1}}}[c] \geq (n^{1+\epsilon'})^{\log_{n^{1+\epsilon'}}n^{1+\hat\epsilon}} )
	\leq
	 kn^{1-\log_{n^{1+\epsilon'}}n^{1+\hat\epsilon}}
  \end{multline*}  
 Let \(y=1-\log_{n^{1+\epsilon'}}n^{1+\hat\epsilon} \), note that for \(\epsilon'<\hat\epsilon \) it holds \(y<0 \). Also let \(R_q=\prob^\sigma_{p\vec{n}}(\{\alpha\mid \textit{ the first state of }M_i\textit{ in }\alpha\textit{ is }q \}) \). Then we can write for each \(\epsilon'<\hat\epsilon<\epsilon \) \begin{multline*}
 	\prob^{\sigma}_{p \vec{n}}(\calC[c] \geq n^{1+\epsilon}\mid \mecs{=}M_1,\dots,M_i )
 \leq \\ \leq
 \prob^{\sigma}_{p \vec{n}}(\calC[c] \geq n^{1+\epsilon'}\mid \mecs{=}M_1,\dots,M_{i-1} )+
 \sum_{q\in Q} R_q 
 \prob^{\sigma'}_{q n^{1+\epsilon'}}(\calC_{M_i^{\bw_{i-1}}}[c] \geq n^{1+\hat\epsilon} )
 \leq
  Z_n+kn^{y}
\end{multline*}
 And since \(\lim_{n\rightarrow \infty} Z_n+kn^{y}=0\), and \(\epsilon',\hat\epsilon \) can be arbitrarily small, we can find values for them for arbitrary \(\epsilon>0 \). Thus \(n \) is a tight estimate of \(\calC[c]\) for type \(M_1,\dots,M_i\).
	
\end{proof}

\section{Proofs for Section~\ref{sec-one-dim}}
\label{app-one-dim}

\begin{lemma}\label{decide-+-in-P}
	Given a one-dimesional VASS MDP, deciding existence of an increasing BSCC $\B$ of $\A_\sigma$ for some $\sigma \in \Sigma_{\MD}$ can be done in \(\PTIME \). 
	\end{lemma}
	\begin{proof}
	If such BSCC exists, then it gives us a solution \(\bx\) with \(\sum_{(p,\bu,q)\in T} \bx((p,\bu,q))\bu(c) > 0 \) for (I). The solution is such that if \(\bw \) is the invariant distribution on states of \(\mathbb{B} \) under \(\sigma\), then for each transition \((p,\bu,q) \) contained in \(\mathbb{B} \), \(\bx((p,\bu,q))=\bw(p) \) if \(p\in Q_n \) is a non-deterministic state, and \(\bx((p,\bu,q))=\bw(p)P((p,\bu,q))  \) if \(p\in Q_p \) is a probabilistic state, while \(\bx(t)=0 \) for each \(t\) that is not contained in \(\mathbb{B} \). And every solution \(\bx\) of (I) such that \(\sum_{(p,\bu,q)\in T} \bx((p,\bu,q))\bu(c) > 0 \) can be used to extract a strategy with expected positive effect on the counter (Appendix~\ref{app-linquad}: Lemma~\ref{app-lemma-solution-I-startegy}). But this is only possible if there exists an increasing BSCC $\B$ of $\A_\sigma$ for some $\sigma \in \Sigma_{\MD}$, as these are the extremal values of any strategy.\footnote{Here we rely on well-known results about finite-state MDPs \cite{Puterman:book}.}
	\end{proof}

	\begin{lemma}\label{lemma-+-exists}
	Given a one-dimensional VASS MDP, if there exists an increasing BSCC $\B$ of $\A_\sigma$ for some $\sigma \in \Sigma_{\MD}$,  then \(\calC[c] \)  and \(\calL \) are unbounded in type \(M\) such that \(\mathbb{B}\subseteq M \). Furthermore, let \(M[t] \) be the MEC containing the transition \(t\). If \(M[t] \) exists and \(\mathbb{B}\subseteq M[t] \), then \(\calT[t] \) is unbounded for type \(M[t] \). Additionally, \(\calT[t] \) is also unbounded for each type \(\beta=M_1,\dots,M_k \) such that there exist \(j\leq i\) such that \(M_i=M[t] \) and \(\mathbb{B}\subseteq M_j \). 
	\end{lemma}
	The computation under \(\sigma \) from any state of \(\mathbb{B}\) has a tendency to increase the counter, and as \(n\) goes towards \(\infty \) the probability of the computation terminating goes to \(0\).\footnote{For formal proof see e.g. \cite{BCKNV:probVASS-linear-termination} (Lemma~6)} Therefore both \(\calC[c] \) and \(\calL \) are unbounded for type \(M \) with \(\mathbb{B}\subseteq M \). Furthermore, if \(\mathbb{B}\subseteq M[t] \), then \(t \) can be iterated infinitely often with high probability by  periodically “deviating”  from \(\sigma\) by temporarily switching to some other strategy which never leaves \(M[t] \) and has positive chance of using \(t\).  Clearly this can be done in such a way that the overall strategy still has the tendency to increase the counter. Therefore in such case \(\calT[t] \) is unbounded for type \(M[t] \). The last part of the theorem comes from the fact that we can first pump the counter in \(M_j \) to an arbitrarily large value, before moving to \(M[t] \) where we then can iterate any strategy on \(M[t] \) that has positive chance of using \(t\).
	

	\begin{lemma}\label{lemma-existence-of-0r}
	Given a one-dimensional VASS MDP, if there exists an unbounded-zero BSCC $\B$ of $\A_\sigma$ for some $\sigma \in \Sigma_{\MD}$, then \(\calL_{\eexp}\in \Omega(n^2) \) and \(n^2 \)  is a lower estimate of \(\calL \) for type \(M \) such that \(\mathbb{B}\subseteq M\). 
	 Furthermore, if \(\mathbb{B} \) contains the transition \(t\), then also \(\calT_{\eexp}[t]\in \Omega(n^2) \) and \(n^2 \)  is a lower estimate of \(\calT[t] \) for type \(M \) such that \(\mathbb{B} \subseteq M\).
	\end{lemma}
	This follows follows directly from the results of \cite{BCKNV:probVASS-linear-termination} (Section~3.3).

	%

	\begin{lemma}\label{lemma-existence-of-0d}
	Given a one-dimensional VASS MDP,  if there exists an bounded-zero BSCC $\B$ of $\A_\sigma$ for some $\sigma \in \Sigma_{\MD}$, then \(\calL\) is unbounded for type \(M\) such that \(\mathbb{B} \subseteq M \). Furthermore, if \(\mathbb{B} \) contains \(t\) then also \(\calT[t]\) is unbounded for type \(M\) such that \(\mathbb{B}\subseteq M \).	
	\end{lemma}
	
	\begin{proof}
	Since \(\mathbb{B} \) is bounded-zero, it must hold that there is no non-zero cycle in \(\mathbb{B} \). Therefore the effect of every path of \(\mathbb{B} \) is bounded by some constant. As such, the computation under \(\sigma \) started from any state of \(\mathbb{B} \) can never terminate if the initial counter value is sufficiently large.
	\end{proof}

	\begin{lemma}\label{lemma-0d-decidable-in-P-when-no-+}
	It is decidable in polynomial time if a one-dimensional VASS MDP \(\A\), that contains no increasing BSCC of $\A_\sigma$ for any $\sigma \in \Sigma_{\MD}$, whether \(\A \) contains a bounded-zero BSCC \(\mathbb{B} \) of $\A_\sigma$ for some $\sigma \in \Sigma_{\MD}$. 
	\end{lemma}
	
	\begin{proof}
	Since there is no class increasing BSCC of an MD strategy, there can be no solution \(\bx\) to (I) with \(\sum_{(p,\bu,q)\in T} \bx((p,\bu,q))\bu(c) > 0 \) as any such solution can be used to extract a strategy with expected positive effect on the counter (Appendix~\ref{app-linquad}: Lemma~\ref{app-lemma-solution-I-startegy}). Therefore from Lemma~\ref{lemma:dichotomy}, we have that there exists a ranking function \(rank\), defined by a maximal solution of (II) (see Section~\ref{section-linquad}), such that the effect of any transition from a nondeterministic state has non-positive effect on \(rank\), and the expected effect of a single computational step taken from a probabilistic state is non-positive on \(rank \). Furthermore, \(rank\) depends on the counter value. Therefore any BSCC which contains any transition whose effect on \(rank\) can be non-zero cannot be bounded-zero. If such transition were from non-deterministic state, then it could only decrease the rank, and as \(rank\) can never be increased in expectation, this would lead to a positive chance of a cycle with negative effect on \(rank\) and thus also on the counter. And if the transition were from a probabilistic state, then as the expectation is non-positive, there would be a non-zero probability of a transition with negative effect on \(rank \) being chosen. Therefore any bounded-zero BSCC can contain only those transitions that never change \(rank\).
	
	On the other hand, any BSCC \(\mathbb{B}\) of $\A_\sigma$ for some $\sigma \in \Sigma_{\MD}$, which contains only transitions that never change \(rank\) must be bounded-zero, as that means the effect of any cycle in \(\mathbb{B} \) must be \(0\) (as any non-zero cycle would have necessarily changed \(rank\) in at least one of its transitions).
	
	Therefore it is sufficient to decide whether there exists a  BSCC \(\mathbb{B}\) of $\A_\sigma$ for some $\sigma \in \Sigma_{\MD}$, containing only those transitions that do not change \(rank \). We can do this by analyzing each MEC one by one. For each MEC we first compute \(rank\) using the system (II) (see Section~\ref{section-linquad}), then proceed by first removing all transitions that can change \(rank\), and then iteratively removing non-deterministic states that do not have any outgoing transition left, and probabilistic states for which we removed any outgoing transition, until we reach a fixed point. If there exists a bounded-zero BSCC \(\mathbb{B}\) of $\A_\sigma$ for some $\sigma \in \Sigma_{\MD}$, then all transitions of \(\mathbb{B} \) will remain in the fixed point as they can never be removed. On the other hand once we reach the fixed point, it holds that for any state \(p\) that is left there either exists a “safe” outgoing transition if \(p\in Q_n\) or all outgoing transitions are “safe” if \(p\in Q_p \), and these “safe” transitions end in a “safe” state. With state being “safe” if it is left in fixed point, and transitions being “safe” if their effect on \(rank \) is \(0\). Thus we can simply select any MD strategy on the states/transitions that are left and it must have a bounded-zero BSCC. And if the fixed point is empty, then there can be no bounded-zero BSCC \(\mathbb{B}\) of $\A_\sigma$ for some $\sigma \in \Sigma_{\MD}$. Clearly this can be done in polynomial time.
	\end{proof}
	
	\begin{note}
	One might ask whether the restriction on one-dimensional VASS MDPs not containing an increasing BSCC of $\A_\sigma$ for some $\sigma \in \Sigma_{\MD}$, is necessary in Lemma~\ref{lemma-0d-decidable-in-P-when-no-+}. The answer is yes, as Lemma~\ref{lemma-np-hard-0d-gen} shows that deciding existence of a bounded-zero BSCC of $\A_\sigma$ for some $\sigma \in \Sigma_{\MD}$, is \(\NP\)-complete for general one-dimensional VASS MDPs. 
	\end{note}

	\begin{lemma}\label{lemma-no-+-or-0d-effect-on-L}
		Given a one-dimensional VASS MDP \(\A\), if there is no increasing or bounded-zero BSCC of \(A_\sigma \) for any \(\sigma\in \Sigma_{\MD} \), then \(n^2\) is an upper estimate of \(\calL\) for every type. 
	\end{lemma}
	
	\begin{lemma}\label{lemma-no-+-effect-on-C}
		Given a one-dimensional VASS MDP \(\A\), if there is no increasing BSCC of \(A_\sigma \) for any \(\sigma\in \Sigma_{\MD} \), then \(n\) is an upper estimate of \(\calC[c]\) for every type. 
	\end{lemma}
	
	To prove these two Lemmata, we need to consider a certain overapproximation of \(\A\), which in some sense is in multiple states at the same time. This overapproximation will allow us to view any computation on \(\A\) as if with very high probability, the computation was at each step \emph{choosing} one of finitely many (depending only on \(\A\)) random walks/cycles, whose effects correspond to their corresponding BSCC (increasing, bounded-zero, unbounded-zero, decreasing). That is if there is no increasing or bounded-zero BSCC of \(A_\sigma \) for any \(\sigma\in \Sigma_{\MD} \), then the expected effect of these random walks can only either be negative (decreasing), or \(0\) but with non-zero variance (unbounded-zero). This then allows us to provide some \emph{structure} to the VASS MDP which will then allow us to prove these lemmata. A key concept to defining this overapproximation is that of an MD-decomposition, which roughly states that for each path on a VASS MDP, we can \emph{color} each transition using one of finitely many colors, such that the sub-path corresponding to each color is a path under some MD strategy associated with that color. We then show that we can \emph{color} any path on \(\A\) using a finite set of colors, and that this \emph{coloring} can be made “on-line”, that is a color can be assigned \emph{uniquely} (in some sense) to each transition at the time this transition is taken in the computation.

	\begin{lemma*}[\textbf{\ref{lemma-MD-decompose}}]
		For every VASS MDP \(\A\), there exist \(\pi_1,\dots,\pi_k \in \Sigma_{\MD}\),  \(p_1,\dots,p_k \in Q\), and a function \(\Decompose_\A \) such that the following conditions are satisfied for every finite path $\alpha$:
		\begin{itemize}
			\item \(\Decompose_\A(\alpha) \) returns an MD-decomposition of $\alpha$ under $\pi_1,\dots,\pi_k$ and $p_1,\ldots,p_k$.
			\item $\Decompose_\A(\alpha)=\Decompose_\A(\alpha_{..\length(\alpha)-1})\; \circ\; \gamma^1\circ\cdots\circ\gamma^k$, where exactly one of 
			$\gamma^i$ has positive length (the $i$ is called the \emph{mode} of $\alpha$).
			\item If the last state of \(\alpha_{..\length(\alpha)-1}\) is probabilistic, then the mode of $\alpha$ does not depend on the last transition of $\alpha$. 
		\end{itemize}
		
\end{lemma*}
	
	\begin{proof}
	Proof by induction on the number of outgoing transitions from non-deterministic states in \(\A\).
	
	Base case: Every non-deterministic state has exactly one outgoing transition. Then there exists only a single strategy \(\pi\) and it is MD. Therefore let \(k=|Q|\), \(\pi_1=\dots=\pi_{k}=\pi \), and \(p_1,\dots,p_k \) be all the distinct states of \(\A \). Then let \(\Decompose_\A(\epsilon)=\epsilon \), and for a path \(\alpha\) with initial state \(p_i\) let \(\Decompose_\A(\alpha)=\Decompose_\A(\alpha_{..\length(\alpha)-1})\circ\gamma^1\circ\gamma^2\circ\dots\circ\gamma^k\) such that \(\gamma^j=p_j \) for \(j\neq i\)  and \(\gamma^i=q,\bu,r \) where \(\alpha=p_i,\dots,q,\bu,r\).   
	
	Induction step: Assume every VASS MDP \(\A'\) with less then \(i \) outgoing transitions from non-deterministic states satisfies the lemma, and let \( \A\) have exactly \(i\) outgoing transitions from non-deterministic states.
	
	If \(\A\) contains no non-deterministic state \(p\in Q_n\) with  \(|\tout(p)|\geq 2 \) then base case applies. Otherwise, let us fix some state \(p\in Q_n\)  with \(|\tout(p)|\geq 2 \), and let \(t_r,t_g\in \tout(p)\) be such that \(t_r\neq t_g \). Let \(\A_r \) and \( \A_g \) be VASS MDPs obtained from \(\A\) by removing \(t_g \) and \(t_r\), respectively. 
	
	For any path \(\alpha \), we define a \emph{red/green-decomposition} of \(\alpha\) on \(\A\) as $\alpha = g_1\circ r_1\circ g_2\circ r_2\circ\dots\circ g_\ell \circ r_\ell$ (all of positive length except potentially \(g_1\) and \( r_\ell\)) satisfying the following:
	
	\begin{itemize}
		\item for every \(1\leq i<\ell \), the last state of $g_i$ is \(p\); 
		\item if \(\length(r_\ell)>0 \) then the last state of \(g_\ell \) is \(p\);
		\item for every \(1\leq i<\ell  \), the last state of $r_i$ is \(p\); 
		\item for every $1<i \leq \ell$, the first state of $g_i$ is \(p\) and the first transition of $g_i$ is \(t_g\);
		\item for every $1\leq i \leq \ell$, the first state of $r_i$ is \(p\)  and the first transition of $r_i$ is \(t_r\);
		\item $g_\alpha=g_1 \circ \cdots \circ g_\ell$ is a path on \(\A_g\). 
		\item $r_\alpha=r_1 \circ \cdots \circ r_\ell$ is a path on \(\A_r\). 
	\end{itemize} 
	
	Clearly every path on \(\A \) has a unique red/green-decomposition that can be computed online.
	
	Now let \(\Decompose_{\A_g}\) and \(\Decompose_{\A_r}\), \(\pi_1^g,\dots, \pi_{k_g}^g \) and \(\pi_1^r,\dots,\pi_{k_r}^r \), \(p_1^g,\dots, p_{k_g}^g \) and \(p_1^r,\dots,p_{k_r}^r \), \(k_g\) and \(k_r \) be the \(\Decompose \) functions, MD strategies, states and \(k\) values for \(\A_g \) and \( \A_r\), respectively. Note that their existence follows from the induction assumption. Then let \(k=k_g+k_r\), \(\pi_1=\pi_1^g,\pi_2=\pi_2^g,\dots, \pi_{k_g}=\pi_{k_g}^g,\pi_{k_g+1}=\pi_{1}^r,\pi_{k_g+2}=\pi_2^r,\dots,\pi_{k_g+k_r}=\pi_{k_r}^r \), \(p_1=p_1^g,p_2=p_2^g,\dots, p_{k_g}=p_{k_g}^g,p_{k_g+1}=p_{1}^r,p_{k_g+2}=p_2^r,\dots,p_{k_g+k_r}=p_{k_r}^r \).  We now define  \(\Decompose_\A(\epsilon)=\epsilon \) and  \(\Decompose_\A(\alpha)=\Decompose_\A(\alpha_{..\length(\alpha)-1})\circ\gamma^1\circ\gamma^2\circ\dots\circ\gamma^k \) such that if \(t=(q,\bu,r) \) is the last transition of \(\alpha \), and \(\alpha = g_1\circ r_1\circ g_2\circ r_2\dots g_\ell \circ r_\ell \) is the {red/green}-decomposition of \(\alpha\) on \(\A\), it holds:
	\begin{itemize}
		\item if \(\length(r_\ell)=0\), then let \(i\) be the \(\Decompose_{\A_g}\)-mode of $g_\alpha=g_1 \circ \cdots \circ g_\ell$. Then we put \(\gamma^j=p_j \) for each \(j\neq i\), and \(\gamma^i=q,\bu,r \); 
		\item if \(\length(r_\ell)>0\), then let \(i\) be the \(\Decompose_{\A_r}\)-mode of $r_\alpha=r_1 \circ \cdots \circ r_\ell$. Then we put \(\gamma^j=p_j \) for each \(j\neq k_g+i\), and \(\gamma^{k_g+i}=q,\bu,r\); 
	\end{itemize}
	\end{proof}


	From Lemma~\ref{lemma-MD-decompose} we can view any strategy \(\sigma\) on $\A$ as if \(\sigma \) were choosing "which of the $k$ MD strategies to advance" at each computational step. That is, let \(\alpha \) be some path produced by a computation under a strategy \(\sigma \), then the "MD strategy to advance" chosen by \(\sigma\) after \(\alpha \) is the MD strategy \(\pi_i \) where \(i\) is such that 
	\begin{itemize}
		\item if the last state \(p\) of \(\alpha \) is probabilistic, then \(i\) is the \(\Decompose_\A\)-mode of the path \(\alpha,\bu,q\), for any \((p,\bu,q)\in \tout(p) \) (note that \(i\) does not depend on which transition of \(\tout(p) \) is chosen);
		\item if the last state \(p\) of \(\alpha \) is non-deterministic,  then let \(t=(p,\bu,q)\in \tout(p) \) be the transition chosen by \(\sigma\) in \(\alpha\). Then  \(i\) is the \(\Decompose_\A\)-mode of the path \(\alpha,\bu,q\).
	\end{itemize}

	Each of these MD strategies \(\pi_i\) can be expressed using a Markov chain \(\M_i\) which is initialized in state \(p_i \). Whenever an MD strategy gets chosen, then the corresponding Markov chain makes one step.	 Naturally, there are some restrictions on which of the indexes can be chosen at a given time, namely a strategy can only choose an index \(i\) such that \(\M_i\) is currently in the same state as the Markov chain which was selected last. However, for our purposes we will consider \emph{pointing strategies} that are allowed to choose any index, regardless of the current situation. We shall call a VASS MDP where such \emph{pointing strategies} are allowed, while also adding a special ``die'' transition that causes instant termination an \emph{extended VASS MDP}.   
	 
	 Formally speaking, let \(\pi_1,\dots, \pi_k \) and \(p_1,\dots, p_k \) be the MD strategies and states from Lemma~\ref{lemma-MD-decompose} associated with \(\Decompose_\A \). An \emph{extended VASS MDP} associated to the 1-dimensional VASS MDP \(\A \) is the 2-dimensional VASS MDP \(\A'= \ce{Q', (Q_n',Q_p'),T',P'}\) where \( Q'=Q^k\times \{0,1,\dots,k \}\), \(Q_n'=Q^k\times \{0 \}\), \(Q_p=Q^k\times \{1,\dots,k \}\), and 
	 \begin{itemize}
	 	\item \(T'=T_n'\cup T_p'\cup T_{die}'\) where
	 	\begin{itemize}
	 		\item \(T_n'=\{((p_1,\dots,p_k,0)),(0,0),(p_1,\dots,p_k,i)\mid (p_1,\dots,p_k)\in Q^k, i\in \{1,\dots,k\}  \}\);
	 		\item \(T_p'= \{((p_1,\dots,p_k,i),(\bu_i,0),(p_1,\dots,p_{i-1},q_i,p_{i+1},\dots,p_k,0))\mid (p_1,\dots,p_k)\in Q^k, i\in \{1,\dots,k\}, \textit{ and either }\pi_i(p_i)=(p_i,\bu_i,q_i)\textit{ or both of } p_i\in Q_p \textit{ and }(p_i,\bu_i,q_i)\in T  \} \);
	 		\item \(T_{die}'=\{(p,(0,-1),p)\mid p\in Q_n'\} \);
	 	\end{itemize} 
	 	\item \(P'(((p_1,\dots,p_k,i),(\bu_i,0),(p_1,\dots,p_{i-1},q_i,p_{i+1},\dots,p_k,0)))= \begin{cases} 
	 		1 &\mbox{if } p_i\in Q_n \\ 
	 		P((p_i,\bu_i,q_i)) & \mbox{if } p_i\in Q_p 
	 	\end{cases} \)
	 \end{itemize}

	
	We call strategies on the extended VASS MDP \emph{pointing strategies}. Note that each strategy on a VASS MDP has an equivalent pointing strategy. Whenever a pointing strategy \(\sigma \) chooses a transition \(((p_1,\dots,p_k,0),(0,0),(p_1,\dots,p_k,i)) \), then we say \(\sigma \) pointed at the Markov chain \(\M_i \). 
	  Note that in the following we only consider computations on the extended VASS MDP initiated in the initial state \((p_1,\dots,p_k,0) \), and with the second counter being set to \(0\), so to simplify the notation, we will write only \(\prob^\sigma_n \) instead of \(\prob^\sigma_{(p_1,\dots,p_k,0)(n,0)} \).

	Given a sequence of strategies  \(\sigma_1,\sigma_2,\dots\) we will define a sequence of pointing strategies \(\sigma_1^\delta,\sigma_2^\delta,\dots \) such that each \(\sigma_n^\delta \) in some sense “behaves as” \(\sigma_n\), but at the same time it “precomputes” the individual Markov chains. 
Since a formal description of \(\sigma_n^\delta \) would be overly complicated, we will give only a high level description of \(\sigma_n^\delta \). The sequence  \(\sigma_1^\delta,\sigma_2^\delta,\dots \) is parameterized by \(0<\delta<1 \). To help us define the behavior of \(\sigma_n^\delta \), we assume \(\sigma_n^\delta \) “remembers” (it can always compute these from the input) some paths \(\gamma_1,\dots,\gamma_k,\alpha\). At the beginning these are all initialized to \(\gamma_1=\dots=\gamma_k=\alpha=\epsilon\). 

A computation under \(\sigma_n^\delta \) operates as follows: First \(\sigma_n^\delta \) internally selects  \(i\in\{1,\dots,k \}\)  that  \(\sigma_n \) would select after \(\alpha \); that is \(i\) is the \(\Decompose_\A\)-mode of \(\alpha'\), where \(\alpha'\) is such that if the last state \(p\) of \(\alpha\) is probabilistic then \(\alpha' \) is \(\alpha\) extended by a single transition, and if \(p \) is nondeterministic then \(\alpha'=\alpha,\bu,q \) is \(\alpha \) extended by the transition \((p,\bu,q) \) where \((p,\bu,q) \) is the transition chosen by \(\sigma_n\) in \(\alpha \) (i.e. \((p,\bu,q)\) is chosen at random using the probabilistic distribution \(\sigma_n(\alpha) \)). Then \(\sigma_n^\delta \) asks if \(\gamma_i\neq \epsilon \), if yes then it skips to step 2), otherwise it first performs step 1) before moving to step 2):
\begin{enumerate}[label=\arabic*)]
	\item Let \((p_1,\dots,p_k,0) \) be the current state of \(\A'\). Then in each non-deterministic state \(\sigma_n^\delta \) keeps pointing at \(\M_i \) until either, if \(p_i\) is not a state of a BSCC of \(M_i \), it reaches a state \((p_1,\dots,p_{i-1},q_i,p_{i+1},\dots,p_k,0) \) where \(q_i \) is a state of a BSCC of \(M_i \) while, or if \(p_i \) is a state of a BSCC of \(\M_i \) then \(\sigma_n^\delta \) stops pointing at \(\M_i \) with probability \(1/2\) each time the computation returns to \((p_1,\dots,p_k,0) \).
	
	  In both cases, if this takes more then $2n^\delta$ steps  then  $\sigma_n^\delta$  terminates using the “die” transitions (i.e. \(\sigma_n^\delta \) keeps reducing the second counter until termination using the transitions from \(T_{die}' \)). After this ends, \(\sigma_n^\delta \) sets \(\gamma_i \) to the path generated by the probabilistic transitions along this iterating (note that this can be seen as a path on \(\M_i \)).
	\item Let \(\gamma_i=p_1,\bu_1,p_2\dots p_\ell \). Since \(\ell>1 \) we have all the information needed to know which index \(\sigma_n \) would have chosen in it's next step. Let \(\alpha'=\alpha,\bu_1,p_2 \) be \(\alpha \) extended by the transition \((p_1,\bu_1,p_2)\). Then \(\sigma_n^\delta \) replaces \(\alpha \) with \(\alpha' \),  and \(\gamma_i \) with the path \(p_2\dots p_\ell \) obtained by removing the first transition from \(\gamma_i \).
\end{enumerate}

	 At this point this process repeats, until \(\alpha \) is a terminating path for initial counter value \(n\) on \(\A \), at which point \(\sigma^\delta_n \) terminates using the transitions from \(T_{die}' \).

	Let \(l=k\cdot u\), where \(u\) is the maximal possible change of the counter per single transition. When started in the initial state \((p_1,\dots,p_k,0) \), we can view \(\sigma_n^\delta \) as if we were computing as per \(\sigma_n \), but occasionally we made some “extra” (precomputed) steps in some of the Markov chains. These “extra” steps correspond exactly to the paths \(\gamma_1,\dots,\gamma_k \), and since probability of these being longer then \(kn^\delta\) decreases exponentially with \(n\), the probability. that \(\sigma_n^\delta \) started with initial counter values \((n+ln^{\delta},0) \) terminates before \(\sigma_n \) would have for initial counter value \(n \) in the first \(n^{2+\epsilon} \) steps, goes to \(0 \) as  \(n\) goes to \(\infty \), for each \(\epsilon>0\). Therefore, if it were to hold that \(\sigma_n \) can perform more then \(n^{2+\epsilon} \) steps with probability at least \(a>0 \) for some \(\epsilon>0 \) and for infinitely many \(n\),  conditioned that  \(\mecs{=} \beta \) for some \(\beta\) (note that \(\beta \) does not depend on \(n\)), then it holds \(\limsup_{n\rightarrow\infty} \prob_{n+ln^{\delta}}^{\sigma^\delta_n}(\calL\geq n^{2+\epsilon} )\geq \frac{a\cdot \weigth(\beta)}{2}>0\). 
	
	Similarly, the counter value of the first counter \(c\) of \(\A' \), when computing under \(\sigma_n^\delta \) from initial value \(n+ln^\delta \) is at each point at most \(n+ln^\delta\) plus the effect of the paths \( \gamma_1,\dots, \gamma_k\), and \(\alpha\). As the length of all of \(\gamma_1,\dots,\gamma_k \) is at most \(kn^\delta \), their total effect on the counter at each point can be at most \(ln^\delta \). And \(\alpha\) is the path generated by a computation of \(\sigma_n \). Therefore, if it were to hold that \(\sigma_n \) can pump the counter to more then \(n^{1+\epsilon} \) with probability at least \(a>0 \) for some \(\epsilon>0 \) and for infinitely many \(n\), conditioned that  \(\mecs{=} \beta \) for some \(\beta\) (note that \(\beta \) does not depend on \(n\)), then it holds \(\limsup_{n\rightarrow\infty} \prob_{n+ln^{\delta}}^{\sigma^\delta_n}(\calC[c]\geq n^{1+\epsilon}-ln^\delta )\geq \frac{a\cdot \weigth(\beta)}{2}>0 \).
	
	Therefore the following two lemmatta imply Lemmata~\ref{lemma-no-+-or-0d-effect-on-L}~and~\ref{lemma-no-+-effect-on-C}.

	
	
	
	
	

	\begin{lemma}\label{lemma-lim-linear-no-+-extended}
		If \(\A \) is a one-dimensional VASS MDP such that there is no increasing BSCC of \(A_\sigma \) for any \(\sigma\in \Sigma_{\MD} \), then \[\limsup_{n\rightarrow\infty} \prob_{n+ln^{\delta}}^{\sigma^\delta_n}(\calC[c]\geq n^{1+\epsilon}-ln^\delta )=0\] for each \(0<\delta<\epsilon<1 \).
	\end{lemma}
	
	\begin{lemma}\label{lemma-extended-limits}
	If \(\A \) is a one-dimensional VASS MDP such that there is no increasing or bounded-zero BSCC of \(A_\sigma \) for any \(\sigma\in \Sigma_{\MD} \), then \[\limsup_{n\rightarrow\infty} \prob_{n+ln^{\delta}}^{\sigma^\delta_n}(\calL\geq n^{2+\epsilon} )= 0 \] for each \(0<\delta<\epsilon<1 \). 
	\end{lemma}
	
	\begin{proof}
		Let us begin with a proof for Lemma~\ref{lemma-extended-limits}. For simplification, let us assume that if the first counter of \(\A' \) becomes negative while iterating in some Markov chain $\M_j$ before it hits the target state of \(\M_j \), that is while performing step 1) as per the description of \(\sigma_n^\delta \), then the computation does not terminate and instead it continues until this target state is reached at which point the computation terminates if the counter is still negative. Clearly this can only prolong the computation. Therefore, each Markov chain $\M_j$ contributes to computation of $\sigma_n^\delta$ by at most a single path \(\alpha_j \) (to reach a BSCC), and then of cycles over some state of a BSCC of \(\M_j \). Let \(X_i^j \) denote the effect of the $i$-th cycle of $\M_j$ performed under the computation of \(\sigma_n^\delta \). As each BSCC of \(A_\sigma \) for any \(\sigma\in \Sigma_{\MD} \) is either unbounded-zero or decreasing,  it holds that either \(\Exp^{\sigma_n^\delta}(X_i^j)= 0 \) while $Var^{\sigma_n^\delta}(X_i^j)>0$ (unbounded-zero), or $\Exp^{\sigma_n^\delta}(X_i^j)<0$ (decreasing).

	It also holds that $\Exp^{\sigma_n^\delta}(\length(\alpha_j))=b_j $ for some constant \( b_j\), and as the length of each \(\alpha_i\) is bounded by \(n^\delta \), the maximal possible effect of all \(\alpha_1,\dots,\alpha_k \) on the counter is \(ln^\delta \). The maximal length of each cycle is bounded by $n^\delta$, therefore we can upper bound the expected length of all cycles of $\M_j$ as \(n^\delta \) times the expected number of such cycles. Clearly the expected number of cycles corresponding to decreasing BSCCs are at most linear as each such cycle moves expectation closer to \(0\), and there is no way to move the expectation away from \(0\) by more then a constant.
	
	To bound the expected number of cycles corresponding to class unbounded-zero BSCCs, we shall use the following lemma that is proven in appendix~\ref{app-randomwalkin-On2-proof}. 
	\begin{lemma}\label{lemma-randomwalkin-On2}
		Let \(\A\) be a one-dimensional VASS MDP, and let \(X_1,X_2,\dots\) be  random variables s.t. each $X_i$ corresponds to the effect of a path on some unbounded-zero BSCC \(\mathbb{B} \) of \(\A_\sigma \) for some \( \sigma\in \Sigma_{\MD}\),  that starts in some state \(p\) of \(\mathbb{B}\) and terminates with probability \(1/2\) every time \(p\) is reached again. Let \(S_0,S_1,\dots \) be defined as \(S_0=0, S_i=S_{i-1}+X_i \), and \(\tau_n\) be a stopping time such that either \(S_{\tau_n}\leq -n \) or \(S_{\tau_n}\geq n \). Then it holds \(\Exp(\tau_n)\in \calO(n^2) \).
	\end{lemma}
	
	 It says that the expected number of cycles before  their cumulative effect exceeds either \(n^{1+\mu} \) or \(-n^{1+\mu} \) is in \(\calO(n^{2+2\mu}) \), for each \(\mu \). Therefore the expected number of cycles upon either effect of \( -n-2ln^{\delta} \) or \(n^{1+\epsilon} \) is in \(\calO(n^{2+2\epsilon})\) for all \(\epsilon>0 \), as for all sufficiently large \(n\)  it holds \(-n^{1+\epsilon}\leq  -n-2ln^{\delta}  \). Therefore the expected length of whole computation, when started in \(n+ln^{\delta} \), and stopped upon either effect of \(-n-ln^{\delta} \) or \(n^{1+\epsilon} \) is in \(\calO(n^{2+2\epsilon}n^\delta)=\calO(n^{2+2\epsilon+\delta}) \). 

	Let \(X_{n,\epsilon}^\delta \) be the random variable encoding the number of steps the computation under \(\sigma_n^\delta \) takes before the effect on counter is either less than \(-n-ln^\delta\), or at least \(n^{1+\epsilon} \), or until \(\sigma_n^\delta \) performs a “die” move, whichever comes first. The above says that \(\Exp^{\sigma_n^\delta}_{n+kn^\delta}(X_{n,\epsilon}^\delta)\leq an^{2+2\epsilon+\delta} \) for some constant \(a\). Furthermore, let \(P_{n,\epsilon}^\delta \) be the probability that the computation under \(\sigma_n^\delta \) reaches effect on counter at least \(n^{1+\epsilon} \) before either hitting effect less then \(-n-ln^\delta\) or performing a “die” move. Note that for \(0<\epsilon'<\epsilon \), it holds \(\prob_{n+ln^{\delta}}^{\sigma_n^\delta}(X_{n,\epsilon}^\delta\geq X_{n,\epsilon'}^\delta) \leq P_{n,\epsilon'}^\delta  \).  Also note that it holds \(\prob_{n+ln^{\delta}}^{\sigma^\delta_n}(\calC[c]\geq n^{1+\epsilon}-ln^\delta )\leq P_{n,\epsilon'}^\delta  \) for each \(0<\epsilon'<\epsilon \) and for all sufficiently large \(n\), as for sufficiently large \(n\) if the counter reaches \(n^{1+\epsilon}-ln^\delta \) then it had to previously reach \(n^{1+\epsilon'} \), as \(n^{1+\epsilon'} \) grows asymptotically slower then \(n^{1+\epsilon}-ln^\delta \). 
	
	Now we shall use the following Lemma that is proven in the Appendix~\ref{app-lemma-lim-for-Pdeltane-proof}.
	
	\begin{lemma}\label{lemma-lim-for-Pdeltane}
		For each \(0<\delta<\epsilon<1 \), it holds \(\lim_{n\rightarrow\infty} P_{n,\epsilon}^\delta=0 \). 
		\end{lemma}
	
	Note that this already implies Lemma~\ref{lemma-lim-linear-no-+-extended}.
	
	To show also Lemma~\ref{lemma-extended-limits}, let us write
	\[
	\prob_{n+ln^{\delta}}^{\sigma_n^\delta}(\calL\geq n^{2+\epsilon}) \leq \prob_{n+ln^{\delta}}^{\sigma_n^\delta}(X_{n,\epsilon}^\delta \geq n^{2+\epsilon}) +P_{n,\epsilon}^\delta
	\]
	
	and for any \(0<\epsilon'<\epsilon \) \[
	\prob_{n+ln^{\delta}}^{\sigma_n^\delta}(X_{n,\epsilon}^\delta \geq n^{2+\epsilon})
	\leq 
	\prob_{n+ln^{\delta}}^{\sigma_n^\delta}(X_{n,\epsilon'}^\delta \geq n^{2+\epsilon})+\prob_{n+ln^{\delta}}^{\sigma_n^\delta}(X_{n,\epsilon}^\delta\geq X_{n,\epsilon'}^\delta) 
	\leq
	 \prob_{n+ln^{\delta}}^{\sigma_n^\delta}(X_{n,\epsilon'}^\delta \geq n^{2+\epsilon})+P_{n,\epsilon'}^\delta
	\]
	and from Markov inequality we  get \[
	\prob_{n+ln^{\delta}}^{\sigma_n^\delta}(X_{n,\epsilon'}^\delta \geq n^{2+\epsilon})
	\leq 
	\frac{an^{2+2\epsilon'+\delta}}{n^{2+\epsilon}}
	\]
	Which gives us 
	\[
	\prob_{n+ln^{\delta}}^{\sigma_n^\delta}(\calL\geq n^{2+\epsilon}) 
	\leq
	 \prob_{n+ln^{\delta}}^{\sigma_n^\delta}(X_{n,\epsilon'}^\delta \geq n^{2+\epsilon})+P_{n,\epsilon'}^\delta +P_{n,\epsilon}^\delta
	\leq
	\frac{an^{2+2\epsilon'+\delta}}{n^{2+\epsilon}} +P_{n,\epsilon'}^\delta +P_{n,\epsilon}^\delta   
	\]
	
	As this holds for each \(0<\epsilon'<\epsilon \), if we put \(\epsilon'=\frac{\epsilon-\delta}{4} \), then \(2\epsilon'+\delta=\frac{\epsilon-\delta}{2}+\frac{2\delta}{2}=\frac{\epsilon+\delta}{2}<\epsilon \) if \(\delta<\epsilon \), and therefore \(\lim_{n\rightarrow\infty}\frac{an^{2+2\epsilon'+\delta}}{n^{2+\epsilon}}=0  \). Therefore it holds\[
	\lim_{n\rightarrow\infty}\prob_{n+ln^{\delta}}^{\sigma_n^\delta}(\calL\geq n^{2+\epsilon}) 
	\leq
	\lim_{n\rightarrow \infty} (\frac{an^{2+2\epsilon'+\delta}}{n^{2+\epsilon}} +P_{n,\epsilon'}^\delta +P_{n,\epsilon}^\delta )=0 
	\]
	\end{proof}

\subsection{Proof of Lemma~\ref{lemma-randomwalkin-On2}}
\label{app-randomwalkin-On2-proof}

\begin{lemma*}[\textbf{\ref{lemma-randomwalkin-On2}}]
	Let \(\A\) be a one-dimensional VASS MDP, and let \(X_1,X_2,\dots\) be  random variables s.t. each $X_i$ corresponds to the effect of a path on some unbounded-zero BSCC \(\mathbb{B} \) of \(\A_\sigma \) for some \( \sigma\in \Sigma_{\MD}\),  that starts in some state \(p\) of \(\mathbb{B}\) and terminates with probability \(1/2\) every time \(p\) is reached again. Let \(S_0,S_1,\dots \) be defined as \(S_0=0, S_i=S_{i-1}+X_i \), and \(\tau_n\) be a stopping time such that either \(S_{\tau_n}\leq -n \) or \(S_{\tau_n}\geq n \). Then it holds \(\Exp(\tau_n)\in \calO(n^2) \).
\end{lemma*}

Let us begin by showing the following technical result.

\begin{lemma}\label{lemma-finding-m-for-MD}
	Let \(\A \) be a one-dimensional VASS MDP. Let \(\mathbb{B} \) be an unbounded-zero BSCC of \(\A_\sigma \) for a strategy \(\sigma\in\Sigma_{\MD}\), and let \( p\) be a state of \(\mathbb{B} \). Let \(X \) denote the random variable representing the effect of a path under \(\sigma\) initiated in \(p\), that ends with probability \(1/2 \) every time the computation returns to \(p\). Then there exists a function \(m:\mathbb{N}\rightarrow \mathbb{N} \) such that \(m(n)\geq 2n \), \(m\in \calO(n) \), and such that for all sufficiently large \(n\) we get for
	\[X' = \begin{cases}
		m  & X\leq -2n \text{ or } X\geq m(n) \\
		X & \text{else}
	\end{cases}  \]
	that it holds \(\Exp^\sigma_p(X')\geq 0 \) and $Var^\sigma_p(X')\geq a$ for some $a>0$ that does not depend on $n$.

\end{lemma}

\begin{proof}


		Since \(\mathbb{B}\) is unbounded-zero, there exists both a positive as well as a negative cycle on \(\mathbb{B}\). Therefore there exists some \(a>0\) such that $\prob^\sigma_p(X\leq -i)\geq a^i$. Also \(X \) is unbounded both from above as well as from below. As every $|Q|$ steps there is non-zero, bounded from below by a constant, probability that we terminate in at most $|Q|$ steps, it holds for each \(i>0 \) that $\prob^\sigma_p(|X|>i)\leq b^i$ for some \(b<1 \). Therefore also $\prob^\sigma_p(X\geq i)\leq b^i$. We claim the lemma holds for any \(m(n)\geq 4n\log_{b} a \).
	
	

	It holds \begin{align*}
	\sum_{i=m(n)}^{\infty} i\prob_{p}^\sigma(X=i)
	\leq
	\sum_{i=m(n)}^{\infty} i\prob_{p}^\sigma(X\geq i)
	\leq
	\sum_{i=m(n)}^{\infty} ib^i=\frac{b^{m(n)}(-bm(n)+b+m(n))}{(b-1)^2} 
\end{align*}

	And if we put in the value \(m(n)=x4n\log_{b} a  \), for \(x\geq 1\) we obtain 
	
	\begin{align*}
	\frac{b^{m(n)}(-bm(n)+b+m(n))}{(b-1)^2}  
	=&
	\frac{b^{x4n\log_{b} a}(-b(x4n\log_{b} a)+b+(x4n\log_{b} a))}{(b-1)^2} 
	=\\&=
	\frac{a^{x4n}(-bx4n\log_{b} a+b+x4n\log_{b} a)}{(b-1)^2} 
\end{align*}

	And furthermore, \[	m(n)(\prob_{p}^\sigma(X\geq m(n))+\prob_{p}^\sigma(X\leq -2n))
	\geq
	  m(n)\prob_{p}^\sigma(X\leq -2n)
	  \geq 
	   m(n)a^{2n}
	   =
	   a^{2n}x4n\log_{b} a 
	\]

	Also, as \(a<1\), it holds for all sufficiently large \(n\) that
	 \[
	 \frac{a^{x4n}(-bx4n\log_{b} a+b+x4n\log_{b} a)}{(b-1)^2} <
	 a^{2n}x4n\log_{b} a 
	 \]
	 
	  Therefore it holds 
	  
	  \begin{align*}
	  	&\Exp^\sigma_p(X')
	  	=
	  	m(n)(\prob_{p}^\sigma(X\geq m(n))+\prob_{p}^\sigma(X\leq -2n)) +\sum_{i=-2n+1}^{m(n)-1} i\prob_{p}^\sigma(X=i) 	
	  	\geq\\ &\geq 
	  	\sum_{i=m(n)}^{\infty} i\prob_{p}^\sigma(X=i) + \sum_{i=-2n+1}^{m(n)-1} i\prob_{p}^\sigma(X=i) 	  
	  \end{align*}
	 
	 And it also holds\begin{align*}
	  	&0=\Exp^\sigma_p(X)
	  	=
	  	\sum_{i=-\infty}^{\infty} i\prob_{p}^\sigma(X=i)
	  	=
	  	\sum_{i=-\infty}^{-2n} i\prob_{p}^\sigma(X=i) + \sum_{i=-2n+1}^{m(n)-1} i\prob_{p}^\sigma(X=i) + \sum_{i=m(n)}^{\infty} i\prob_{p}^\sigma(X=i) 
	  	\leq \\&\leq 
	  	 	\sum_{i=m(n)}^{\infty} i\prob_{p}^\sigma(X=i) + \sum_{i=-2n+1}^{m(n)-1} i\prob_{p}^\sigma(X=i) 
	  \end{align*}	
  
  And therefore \[\Exp^\sigma_p(X')\geq \sum_{i=m(n)}^{\infty} i\prob_{p}^\sigma(X=i) + \sum_{i=-2n+1}^{m(n)-1} i\prob_{p}^\sigma(X=i)\geq 0 \]

	For the part about $Var(X')$.	Since \(\mathbb{B} \) is unbounded-zero, it holds that $\Exp^\sigma_p(X)^2=Var^\sigma_p(X)\geq y>0$ for some $y$. Therefore it holds for each \(n\) that	\[0<y\leq \Exp^\sigma_p(X^2 )= \sum_{i=1}^{\infty} i^2\prob_p^\sigma(|X|=i) =   \sum_{i=1}^{m(n)} i^2\prob_p^\sigma(|X|=i) +\sum_{i=m(n)}^{\infty} i^2\prob_p^\sigma(|X|=i)\leq  \]
	\[\leq
	\sum_{i=1}^{m(n)} i^2\prob_p^\sigma(|X|=i) +\sum_{i=m(n)}^{\infty} i^2 \prob_p^\sigma(|X|\geq i) 
	\leq 
	\sum_{i=1}^{m(n)} i^2\prob_p^\sigma(|X|=i)  +\sum_{i=m(n)}^{\infty} i^2b^i 
	\]
	And \[
	\sum_{i=m(n)}^{\infty} i^2b^i =
	\frac{b^{m(n)} (m^2(n) (-b^2) + 2 m^2(n) b - m^2(n) + 2 m(n) b^2 - 2 m(n) b - b^2 - b)}{(b - 1)^3}
	\]
	
	But this fraction is dominated by \(b^{m(n)} \) which decreases exponentially in \(n\) (as \(b<1\)). Therefore  for all sufficiently large \(n\) it holds \(y/2 \leq \sum_{i=1}^{m(n)} i^2\prob_p^\sigma(|X|=i) \). But this gives us $\Exp^\sigma_p((X')^2)\geq \sum_{i=1}^{m(n)} i^2\prob_p^\sigma(|X|=i)\geq b/2 $ for each \(m(n)\geq n \) and all sufficiently large \(n\).
\end{proof}


Let us now restate the Lemma~\ref{lemma-randomwalkin-On2}.

\begin{lemma*}[\textbf{\ref{lemma-randomwalkin-On2}}]
Let \(\A\) be a one-dimensional VASS MDP, and let \(X_1,X_2,\dots\) be  random variables s.t. each $X_i$ corresponds to the effect of a path on some unbounded-zero BSCC \(\mathbb{B} \) of \(\A_\sigma \) for some \( \sigma\in \Sigma_{\MD}\),  that starts in some state \(p\) of \(\mathbb{B}\) and terminates with probability \(1/2\) every time \(p\) is reached again. Let \(S_0,S_1,\dots \) be defined as \(S_0=0, S_i=S_{i-1}+X_i \), and \(\tau_n\) be a stopping time such that either \(S_{\tau_n}\leq -n \) or \(S_{\tau_n}\geq n \). Then it holds \(\Exp(\tau_n)\in \calO(n^2) \).
	\end{lemma*}

\begin{proof}
	As there are only finitely many BSCCs of \(\A_\sigma \) for \(\sigma\in\Sigma_{\MD}\), and each of them has only finitely many states, there are only finitely many distributions \(D_1,\dots,D_x \) such that each \(X_i\approx D_y \) for some \(1\leq y \leq x \).
	
	Let \(X_1^n,X_2^n,\dots\) be random variables such that

	\[X_i^n = \begin{cases}
		m(n)  & X_i\leq -2n \text{ or } X_i\geq m(n) \\
		X_i & \text{else}
	\end{cases}  \]
	
	where \(m(n)=an \) is the maximal value of \(m(n) \) obtained from Lemma~\ref{lemma-finding-m-for-MD} for any unbounded-zero BSCC of any \(\A_\sigma \) for any \(\sigma\in \Sigma_\MD \), and \( a\) is some constant. Then it holds that \(\Exp(X_i^n)\geq 0 \), and there exists \(b>0 \) that does not depend on \(n\) such that \(\Exp((X_i^n)^2)\geq b \) for each \( i\). 

	Let \(S_0^n,S_1^n,\dots \) be a random walk defined as \(S_0^n=2n, S_i^n=S_{i-1}^n+X_i^n \), and let \(\tau_n'\) be a stopping time such that either \(S_{\tau_n'}^n\leq 2n-n \) or \(S^n_{\tau_n'}\geq 2n+n \). Clearly it holds that \(\tau_n'=\tau_n \), therefore it is enough to show that \(\Exp(\tau_n')\in \calO(n^2) \). 
	
	Let us proceed by showing the following. 
	
	
	%
	%
		%
		%

	\begin{lemma}
		Let \(M_i^n=(S_i^n)^2 - bi \). Then  \(M_0^n,M_1^n,\dots \) is a submartingale.
	\end{lemma}
	
	\begin{proof}
	
		\begin{align*}
			&\Exp(M^n_{i+1} \mid X_{i}^n,\dots,X_1^n)
			=
			\Exp((S_{i+1}^n)^2 -b(i+1) \mid X_{i}^n,\dots,X_1^n)
			\\&=
			\Exp((S_{i}^n+X_{i+1}^n)^2 -b(i+1) \mid X_{i}^n,\dots,X_1^n)
			\\&=
			\Exp((S_{i}^n)^2+2S_{i}^nX_{i+1}^n + (X_{i+1}^n)^2 -b(i+1) \mid X_{i}^n,\dots,X_1^n)
			\\&=
			(S_{i}^n)^2+2S_{i}^n\Exp(X_{i+1}^n\mid X_{i}^n,\dots,X_1^n))  + \Exp((X_{i+1}^n)^2\mid X_{i}^n,\dots,X_1^n) -b(i+1) 
			\\ &\geq
			(S_{i}^n)^2+0  + b -b(i+1) 		
			=
			(S_{i}^n)^2 +b - bi -b
			= 
			(S_{i}^n)^2 - bi
			=
			M_i
		\end{align*}
	\end{proof}

	
	As it holds that \(\Exp(\tau_n')<\infty \), from the optional stopping theorem we obtain 
	\(\Exp(M_0^n)\leq \Exp(M^n_{\tau_n'}) \) which can be rewritten as \((2n)^2 \leq \Exp((S^n_{\tau_n'})^2 -b\tau_n')  =\Exp((S_{\tau_n'}^n)^2) -b\Exp(\tau_n') \).
	As it holds \((S^n_{\tau_n'})^2\leq (3n+m(n))^2= (3n+an)^2= (9+6a+a^2)n^2 \) this gives us \( 4n^2+b\Exp(\tau_n')\leq \Exp((S^n_{\tau'})^2)\leq (9+6a+a^2)n^2  \) and  so \(\Exp(\tau_n')\leq \frac{(5+6a+a^2)n^2}{b}\in \calO(n^2) \).

\end{proof}

%
%
%
%
%
%
%
%

\subsection{Proof of Lemma~\ref{lemma-lim-for-Pdeltane}}
\label{app-lemma-lim-for-Pdeltane-proof}

\begin{lemma*}[\textbf{\ref{lemma-lim-for-Pdeltane}}]
	For each \(0<\delta<\epsilon<1 \), it holds \(\lim_{n\rightarrow\infty} P_{n,\epsilon}^\delta=0 \). 
\end{lemma*}

Assume there exist some \(0<\delta<\epsilon \) such that \(\limsup_{n\rightarrow\infty } P_{n,\epsilon}^\delta=a>0 \). Then for each \(n_0\), there exists \(n>n_0 \) such that \(P_{n,\epsilon}^\delta>a/2 \). Most notably, this means that the effect of the path \(\alpha \) in \(\sigma_n^\delta \) (see definition of \(\sigma_n^\delta \)) is at least \(n^{1+\epsilon}-ln^\delta \) with probability at least \(a/2 \). But as \(\alpha\) can be equally seen as a path under \(\sigma_n \), this means that also the strategy \(\sigma_n \) reaches effect \(n^{1+\epsilon}-ln^\delta \) before the effect \(-n \) with probability \(R_{n,\epsilon}^\delta \geq a/2 \), for infinitely many \(n\). 
Let type \(\beta_n \) be some type with the largest \(\weigth(\beta) \) among all types \(\beta \), such that with probability at least \(a/2\) \(\sigma_n \) reaches the effect at least \(n^{1+\epsilon}-ln^\delta\) before the effect \(-n\) conditioned the computation follows \(\beta\). If the length of \(\beta_n \) were dependent on \(n\) then as probability of all long types decreases exponentially fast with their length, it could not hold that \(R_{n,\epsilon}^\delta>a/2 \) for arbitrarily large \(n\). Therefore there must exist infinitely many values \(n_1,n_2,\dots\) such that \(\beta_{n_1}=\beta_{n_2}=\dots \), let us denote this type by \(\beta=M_1,\dots,M_x\) (i.e., \(\beta=\beta_{n_1} \)).

This means that \(n \) is not an upper estimate of \(\calC[c] \) for type \(\beta\). But in the next Lemma we are going to show that \(n \) is an upper estimate of \(\calC[c] \) for type \(\beta\), thus showing a contradiction.

\begin{lemma}
	For each \(0<\epsilon_1 \) there exists \(0<\epsilon_2 \) and \(0<b\) such that \[\prob^{\sigma_n}_{p \bn}(\calC[c] \geq n^{1+\epsilon_1}\mid \mecs{=} \beta  )\leq bn^{-\epsilon_2}\] for each state \(p\) of \(M_1\).
	\end{lemma}
\begin{proof}
We are going to do an induction over \(1\leq i\leq x \).

Base case: \( i=1\), then from Lemma~\ref{lemma-linear-upper-bound} we   have \[\prob^{\sigma_n}_{p \bn}(\calC[c] \geq n^{1+\epsilon_1} \mid\mecs{=}M_1)\leq bn^{-\epsilon_1}  \] for some constant \(b\) and for each \(0<\epsilon_1 \). 

Induction step: Assume this holds for \(i<x \), let us now show it holds for \(i+1\) as well.
 From induction assumption we have that for each \(0<\epsilon_1' \) there exists \(0<\epsilon_2' \) and \(0<b'\) such that \(\prob^{\sigma_n}_{p \bn}(\calC[c] \geq n^{1+\epsilon_1'}\mid \mecs{=} M_1,\dots,M_i  )\leq b'n^{-\epsilon_2'} \) . Therefore, when the computation reaches \(M_{i+1} \), the counter is larger than \( n^{1+\epsilon_1'} \) with probability at most \(b'n^{-\epsilon_2'} \). As such we can express for each \(0<\epsilon_1'<\epsilon \) 
 \begin{align*}
 	&\prob^{\sigma_n}_{p \bn}(\calC[c] \geq n^{1+\epsilon}\mid \mecs{=} M_1,\dots,M_{i+1}  )
 	\leq \\ &\leq
 	\prob^{\sigma_n}_{p \bn}(\calC[c] \geq n^{1+\epsilon_1'}\mid \mecs{=} M_1,\dots,M_i  )+
 	\sum_{r\in M_{i+1}}
 	P_q
 	\prob^{\sigma_n^r}_{r n^{1+\epsilon_1'}}(\calC[c] \geq n^{1+\epsilon}\mid \mecs{=} M_{i+1})
 	=\\&=
 	\prob^{\sigma_n}_{p \bn}(\calC[c] \geq n^{1+\epsilon_1'}\mid \mecs{=} M_1,\dots,M_i  )+
 	\prob^{\sigma_n^{q}}_{q n^{1+\epsilon_1'}}(\calC[c] \geq n^{1+\epsilon}\mid \mecs{=} M_{i+1})
 \end{align*}

 where \(\sigma_n^r \) is the strategy which computes as if \(\sigma_n \) after the path from \(p\) to \(r\) for each \(r\) being a state of \(M_{i+1} \), \(P_r=\prob^{\sigma_n}_{p\bn}(\{\alpha\mid \textit{the first state of }M_{i+1}\textit{ in }\alpha\textit{ is }r\}) \), and \(q \) is the state of \(M_{i+1} \) such that for each state \(r\) of \(M_{i+1} \) it holds \[\prob^{\sigma_n^{r}}_{r n^{1+\epsilon_1'}}(\calC[c] \geq n^{1+\epsilon}\mid \mecs{=} M_{i+1})\leq \prob^{\sigma_n^{q}}_{q n^{1+\epsilon_1'}}(\calC[c] \geq n^{1+\epsilon}\mid \mecs{=} M_{i+1})\]
 
  But from Lemma~\ref{lemma-linear-upper-bound} we have that \begin{align*}
 	&\prob^{\sigma_n^q}_{q n^{1+\epsilon_1'}}(\calC[c] \geq n^{1+\epsilon}\mid \mecs{=} M_{i+1})
 =\\&=
 \prob^{\sigma_n^q}_{q n^{1+\epsilon_1'}}(\calC[c] \geq (n^{1+\epsilon_1'})^{\log_{n^{1+\epsilon_1'}}n^{1+\epsilon}}\mid \mecs{=} M_{i+1})
 \leq 
b(n^{1+\epsilon_1'})^{1-\log_{n^{1+\epsilon_1'}}n^{1+\epsilon}}
\end{align*}
 
 Let \(y=1-\log_{n^{1+\epsilon_1'}}n^{1+\epsilon} \), note that \(y<0 \) since \(\epsilon_1'<\epsilon \). Then we can write \[
 \prob^{\sigma_n}_{p \bn}(\calC[c] \geq n^{1+\epsilon}\mid \mecs{=} M_1,\dots,M_{i+1}  )
 \leq 
 b'n^{-\epsilon_2'}+
 bn^{y}
 \]
 
 which for each \(\epsilon_1>\epsilon\) gives \begin{align*}
 	&\prob^{\sigma_n}_{p \bn}(\calC[c] \geq n^{1+\epsilon_1}\mid \mecs{=} M_1,\dots,M_{i+1}  )
 \leq\\& \leq
 \prob^{\sigma_n}_{p \bn}(\calC[c] \geq n^{1+\epsilon}\mid \mecs{=} M_1,\dots,M_{i+1}  )
 \leq 
 b'n^{-\epsilon_2'}+
 bn^{y}
\end{align*}

 And for \(\epsilon_2=\min (\epsilon_2',-y) \) and \(\hat{b}=\max(2b',2b)\) this gives use \[
 \prob^{\sigma_n}_{p \bn}(\calC[c] \geq n^{1+\epsilon_1}\mid \mecs{=} M_1,\dots,M_{i+1}  )
 \leq
 b'n^{-\epsilon_1'}+
 bn^{y}
 \leq 
 \hat{b}n^{-\epsilon_2}
 \]
 thus the induction step holds.
 \end{proof}
%
%
%
%
%
%
%
%
%
%
%
%
%
\subsection{Proof of Lemma~\ref{lem-energy}}
\label{app-lem-energy-proof}

\begin{lemma*}[\textbf{\ref{lem-energy}}]
	An energy MDP has a safe configuration iff there exists a non-decreasing BSCC $\B$ of $\A_\sigma$ for some $\sigma \in \Sigma_{\MD}$. 
\end{lemma*}
\begin{proof}
	The \(\Leftarrow \) direction is trivial. For the \(\Rightarrow\) direction assume the opposite. Then there exists a safe configuration, and a strategy  such that the counter never decreases below some bound. But then from Lemma~\ref{lemma-MD-decompose} we can view the strategy as if choosing which of the finitely many Markov chains is to advance. And since there is no non-decreasing BSCC $\B$ of $\A_\sigma$ for any $\sigma \in \Sigma_{\MD}$, each of these Markov chains contains a negative cycle. Therefore after every at most finite number of steps the counter has non-zero probability of decreasing, thus it cannot be bounded from below for the entire computation. 
\end{proof}

\subsection{Proof of Lemma~\ref{lemma-type-NONNEGATIVE-NP-complete}}
\label{app-lemma-type-NONNEGATIVE-NP-complete-proof}

\begin{lemma*}[\textbf{\ref{lemma-type-NONNEGATIVE-NP-complete}}]
	 The problem whether there exists a non-decreasing BSCC $\B$ of $\A_{\sigma}$ for some $\sigma \in \Sigma_{\MD}$ such that $\B$ contains a given state $p \in Q$ is $\NP$-complete.
	\end{lemma*}
\begin{proof}
	This problem being in \(\NP\) is easy as we simply have to guess a BSCC of some MD strategy, and then verify that it contains no negative cycle while containing \(p\). For the \(\NP\)-hardness let us show a reduction from the \(\NP\)-complete  problem of deciding whether a given graph \(G\) contains a Hamiltonian cycle.
	
	Let \(G=(V,E)\) be the graph for which we want to decide existence of a Hamiltonian cycle, and let \(p\in V\) be one of it's vertices. 
	
	Let \(\A \) be a 1-dimensional VASS MDP, whose set of states is \(V\), all states are nondeterministic, and the set of transitions is \(T\) such that whenever there is an edge \(\{q,r\}\in E, q\neq p\neq r \), then \(\A\) contains the transitions \((q,+1,r),(r,+1,q) \), and for each edge \(\{p,q\}\in E \), \(\A\) contains the transitions \((q,+1,p),(p,-|V|+1,q) \)
	
	We now claim that \(G\) contains a Hamiltonian path iff there exists a non-decreasing  BSCC \(\mathbb{B}\) of \(\A_\sigma \) for some \(\sigma\in\Sigma_{\MD} \) such that \(\mathbb{B}\) contains \(p\).
	
	First let a Hamiltonian cycle \(\alpha=p_1,t_1,p_2,\dots,p_l,t_l,p_1\) exist. Then for the MD strategy \(\sigma(p_j)=t_j \), \(\A_\sigma \) surely contains exactly one BSCC that contains \(p\), and it contains exactly one cycle whose effect is \(0\). Thus it is non-decreasing.
	
	Now let there exists a non-decreasing BSCC \(\mathbb{B}\) of \(\A_\sigma \) for some \(\sigma\in\Sigma_{\MD} \) such that \(\mathbb{B}\) contains \(p\). Then since the effect of every outgoing transition of \(p\) is \(-|V|+1 \), the effect of every other transition is \(+1\), and \(\mathbb{B} \) contains no negative cycles, there must be at least \(|V| \) transitions in \(\mathbb{B} \). But as \(\sigma\) is an MD strategy, there can be at most one transition per state, and so \(\mathbb{B} \) must contain every single state of \(G\). But this means that the computation under \(\pi\) follows a Hamiltonian cycle.		
\end{proof}

\begin{lemma}\label{lemma-np-hard-0d-gen}
	The problem whether there exists a bounded-zero BSCC $\B$ of $\A_{\sigma}$ for some $\sigma \in \Sigma_{\MD}$ is $\NP$-complete for general one-dimensional VASS MDPs.
	\end{lemma}
This follows from the proof above of the previous Lemma, as any bounded-zero BSCC  of $\A_{\sigma}$ for some $\sigma \in \Sigma_{\MD}$ in the VASS MDP \(\A \) constructed for the graph \(G\) must contain \(p \), while the BSCC associated to the strategy obtained from a hamiltonian cycle is bounded-zero.

\end{document}